\documentclass[11pt]{amsart}
\usepackage{bm}
\usepackage{amsmath,amsfonts, amsthm, amssymb}
\usepackage{multirow}
\usepackage{microtype}
\usepackage{enumerate}
\usepackage{graphicx}
\usepackage{tikz}
\usepackage[style=alphabetic, backend=biber]{biblatex}
\graphicspath{{../../figures/}}
\usepackage[onehalfspacing]{setspace}

\DeclareFontFamily{U}{mathx}{\hyphenchar\font45}
\DeclareFontShape{U}{mathx}{m}{n}{
      <5> <4> <7> <8> <9> <10>
      <10.95> <12> <14.4> <17.28> <20.74> <24.88>
      mathx10
      }{}
\DeclareSymbolFont{mathx}{U}{mathx}{m}{n}
\DeclareFontSubstitution{U}{mathx}{m}{n}
\DeclareMathAccent{\widecheck}{0}{mathx}{"71}

\usepackage{xcolor}
\definecolor{MidnightBlue}{RGB}{25,25,150}
\definecolor{BrickRed}{RGB}{182,50,28}
\definecolor{ForestGreen}{RGB}{34,139,34}

\usepackage{scalerel,stackengine}

\def\bra#1{\mathinner{\langle{#1}|}}
\def\ket#1{\mathinner{|{#1}\rangle}}
\def\braket#1{\mathinner{\langle{#1}\rangle}}

\def\Real{{\mathbb R}}
\def\Sphere{{\mathbb S}}
\def\Complex{{\mathbb C}}

\def\Prob{{\mathbf P}}

\def\Expec{{\mathbf E\,}}

\def\eps{{\varepsilon}}

\def\bbT{{\mathbb{T}}}

\def\lsim{{\,\lesssim\,}}

\def\bbZ{{\mathbb{Z}}}
\def\Z{{\mathbb{Z}}}
\def\R{{\mathbb{R}}}
\def\bbN{{\mathbb N}}

\newcommand{\Ft}[1]{{\widehat{#1}}}

\newcommand{\diff}{\mathop{}\!\mathrm{d}}

\numberwithin{equation}{section}

\DeclareMathOperator{\Id}{Id}

\newcommand{\mcal}[1]{{\mathcal{#1}}}

\DeclareMathOperator{\supp}{supp}

\DeclareMathOperator{\trace}{tr}

\usepackage{thmtools, thm-restate}
\declaretheorem[numberwithin=section]{theorem}
\declaretheorem[numberlike=theorem]{proposition}
\declaretheorem[numberlike=theorem]{lemma}

\declaretheorem[numberlike=theorem]{corollary}

\title[Tail bounds for random Dyson series]{Tail bounds for the Dyson series of random Schr\"{o}dinger equations}
\author{Adam Black}
\address{Max Planck Institute for Mathematics in the Sciences, Leipzig, Germany}
\email{adam.black@mis.mpg.de}
\author{Reuben Drogin}
\address{Department of Mathematics, Yale University, New Haven, CT}
\email{reuben.drogin@yale.edu}
\author{Felipe Hern\'{a}ndez}
\address{Department of Mathematics, Penn State University, State College, PA}
\email{felipeh@psu.edu}

\begin{document}

\begin{abstract}
We study Schr{\"o}dinger equations on $\mathbb{Z}^d$ and $\mathbb{R}^d$, $d\geq 2$ with random potentials of strength $\lambda$.  Our main result gives tail bounds for the terms of the Dyson series that are effective at time scales on the order of $\lambda^{-2+\varepsilon}$.  As corollaries, we obtain estimates on the frequency localization and spatial delocalization of approximate eigenfunctions in the spirit of \cite{SSW,Chen}.  These estimates also apply to Floquet states associated to time-periodic potentials.  Our proof is elementary in that we use neither sophisticated harmonic analysis as in \cite{SSW} nor diagrammatic arguments as in \cite{Chen}.  Instead, we use only the noncommutative Khintchine inequality from random matrix theory combined with pointwise dispersive estimates for the free Schr{\"o}dinger equation.
\end{abstract}

\maketitle

\section{Introduction}
In this paper, we consider random Schr\"odinger equations of the form
\begin{align}\label{eq:schro}
i\partial_t\psi = H_0\psi + \lambda V(x,t) \psi,
\end{align}
where $H_0$ is the Laplacian (either on $\bbZ^d$ or $\Real^d$),
$V$ is a random (and possibly time-dependent) potential that we specify
below, and $\lambda>0$ is a small coupling constant. By now, the mathematical study of random Schr\"{o}dinger equations has a long and storied history, though a comprehensive understanding of the small $\lambda$ regime for $d\geq 2$ remains elusive. Typically, the best available results either combine harmonic analysis with entropy methods \cite{SSW,Bourgain} or involve intricate diagramatics \cite{Chen,erdHos2007quantum,erdHos2008quantum,hernandez2024quantum}.

The purpose of this paper is to explain how techniques from random matrix theory give a simple and unified approach to the study of \eqref{eq:schro} at small coupling. In doing so, we recover or sharpen many classical results as well as provide new applications.
Here is a sample of what we can show, and see Section \ref{sec:RelatedWork} for an in-depth comparison to the literature:

\begin{itemize}
	\item We give tail bounds for the quantity $\|U(t,0)-e^{-itH_0}\|_{L^2\to L^2}$, where $U(t,0)$ is the propagator of \eqref{eq:schro} from $t=0$. Roughly, we show that (up to log losses) its mean is of order $\lambda \sqrt{t}$ and its fluctuations are also bounded by $\lambda\sqrt{t}$.\footnote{One can apply similar methods to improve estimates for the higher order moments, but this is deferred to a future work.} An energy restricted version of a similar tail bound appeared previously in work of Bourgain \cite{Bourgain}, albeit for a model with a spatially decaying random potential.
	\item For time-independent $V$, we prove bounds on the frequency localization of eigenfunctions in the spirit of \cite{SSW,Chen}. Our results apply on both $\Z^{d}$ and $\Real^{d}$, $d\geq 2$ and at all energies.  A small modification to the proof shows that when $V$ has periodic time dependence we obtain an analogous frequency localization bound for Floquet states.  This is illustrated in Figure~\ref{fig:floquet}, and appears to be a new result.
	\item We show that, with high probability, the propagator $e^{-itH}$ is well approximated in operator norm by
the truncation to roughly $\lambda^2t$ terms of the Dyson series.
\end{itemize}
In the study of quantum diffusion an important step is to show that the Dyson series can be
truncated, and so far all proofs of this at the timescale $t>\lambda^{-2}$
involve difficult diagrammatic arguments~\cite{erdHos2008quantum,hernandez2024quantum}.  We can show (see Corollary~\ref{cor:close-to-free}) that this truncation can be done with high probability, but our tail bounds do not
suffice to control $\Expec \|e^{-itH} - \sum_{j=0}^M T_j\|_{\textrm{op}}^2$ which is required to estimate the
evolution of observables.  See, however, the discussion after Theorem~\ref{thm:main-RMT} below.
\\

\begin{figure}
\centering
\includegraphics[scale=0.5]{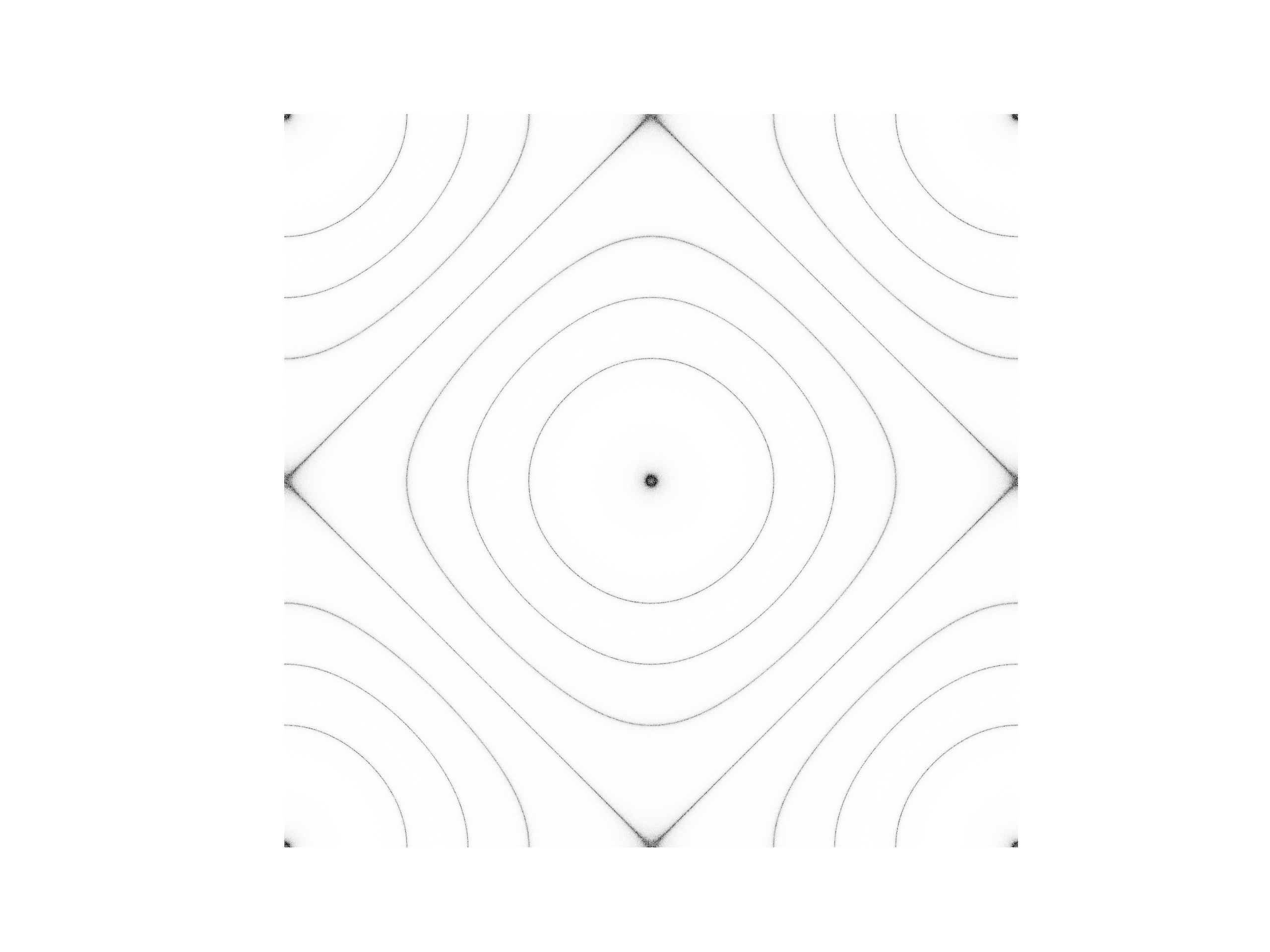}
\caption{
The Fourier transform of an eigenfunction with eigenvalue approximately $1$
for the evolution $U(4\pi,0)$ in a periodically driven system with $V(x,t) = \cos(t/2) V(x)$. Here $V(x)$ has independent Bernoulli $\pm1$ entries on a
$2048\times 2048$ periodic lattice and $H_0$ is the shifted discrete Laplacian~\eqref{eq:HZd}.
The darker pixels represent the Fourier coefficients
with larger magnitude inside the square $[-\frac12,\frac12]^2$.
Note that the mass is concentrated near the level sets
$\cos(2\pi k_x) + \cos(2\pi k_y) \in \frac{1}{2}\bbZ$,
as described by Corollary~\ref{corTimeDependent}.}
\label{fig:floquet}
\end{figure}

\subsection{Model and Main Results: time-independent potentials}
Let us now specify the exact models we wish to consider in the time-independent setting.
On $\bbZ^d$, $d\geq 2$, we take $H=H_0+\lambda V$ where $H_0$ is the discrete Laplacian on $\bbZ^{d}$ normalized as
\begin{align}
\label{eq:HZd}
	(H_0 \phi)(n)=\sum_{|n'-n|=1}\phi(n'),
\end{align}
and $V$ is a potential of the form
\begin{equation}
\label{eq:VR-Zd}
	V=\sum_{\substack{n\in \Z^{d}\\|n|\leq R}}g_n \ket{n}\bra{n},
\end{equation}
with randomness $g_n$ and a cutoff radius $R > \lambda^{-1}$.
Here $g_n$ are either independent standard Gaussian random variables or are independent variables satisfying $\Expec g_n =0$ and $|g_n|\leq 1$ almost surely.

On $\R^d$, $d\geq 2$ will instead set
\begin{equation}
\begin{split}
\label{eq:modelRd}
	&H_0=-\frac{1}{2}\Delta\\
	&V = \sum_{\substack{n\in\bbZ^d \\ |n|\leq R}} g_n V_0(x-n),
\end{split}
\end{equation}
where $V_0\in L^\infty(\Real^d)$ is a fixed bounded function with support in the
ball $B_1$.  In Section \ref{sec:Time-periodic} below, we also consider the Schr\"{o}dinger equation with time-dependent
and periodic potential
\begin{equation}\label{schroe}
	i\partial_t \psi = H_0\psi + \lambda  V(x,t) \psi,
\end{equation}
where $V(x,t+\tau)=V(x,t)$ for some  $\tau$.

We take $H_0$ to be the Laplacian for the sake of simplicity, but use only a few of its dispersive properties (see Lemmas~\ref{lem:nonstationary} and \ref{lem:free-ests-Rd}). Thus, our proof could be adapted to other choices of background Hamiltonian provided that one has an adequate understanding of its unitary evolution.

Our main result concerns the Dyson series of $H$, which we write as
\[
e^{-itH} = \sum_{j=0}^\infty (-i\lambda)^j e^{-itH_0} T_j(t),
\]
where $T_0(t) = \Id$ and for $j\geq 1$, $T_j(t)$ has the form
\begin{align}\label{eq:Tj-def}
T_j(t) = \int_{0\leq s_1\leq \cdots\leq s_j\leq t} V(s_j)V(s_{j-1})\cdots V(s_1)\diff \vec{s},
\end{align}
with $V(s) := e^{isH_0}Ve^{-isH_0}$.

For time-independent $V$, our main theorem is as follows:
\begin{theorem}\label{thm:thmMain}
Let $H$	be the random Schr\"{o}dinger operator defined above on either $\Z^d$ or $\R^d$ with $d\geq 2$. Then there exists a constant $c>0$ independent of $\lambda$ or $j$ so that for all $K>\sqrt{\log(R)}$
\begin{align}\label{eq:mainT_jEst}
	\Prob(\|T_j(t)\|_{\mathrm{op}} \geq (K^2\sigma_d(t)/j)^{j/2}
\text{ for some }t>0,j\in\bbN) \leq 2\exp(-cK^2),
\end{align}
where
\[
\sigma_d(t)=
\begin{cases}  |t|\log(2+|t|), d>2 \\
 |t| \log^2(2+|t|), d=2.
\end{cases}
\]
\end{theorem}
By summing the Dyson series, this implies that $e^{-itH}$ is well-approximated by $e^{-itH_0}$ with
high probability for times $t=O(\lambda^{-2+\eps})$:
\begin{corollary}\label{cor:close-to-free}
In the notation of Theorem \ref{thm:thmMain}, for all $K>\sqrt{\log(R)}$
\begin{equation}
\label{eq:close-to-free}
\Prob(\|e^{-itH} - e^{-itH_0}\|_{\mathrm{op}} \geq K \lambda \sigma_d^{1/2}(t)
\text{ for some }t\in\Real) \leq 2\exp(-cK^2).
\end{equation}
Moreover, for $M>CK^2 \lambda^2 \sigma_d(t)$,
\begin{equation}
\label{eq:truncation}
\Prob(\|e^{-itH} - \sum_{j=0}^M T_j\|_{\mathrm{op}} \geq \exp(-M)) \leq \exp(-cK^2).
\end{equation}
\end{corollary}

Note that such an operator norm estimate cannot be true for the full Anderson Hamiltonian. Without restricting the support of the potential,
one can always find a large ball which resembles any deterministic potential with probability $1$. For $R=\lambda^{-N}$ and $t=O(\lambda^{-2})$, this
is not a serious restriction when considering suitably localized initial data.

\par

The comparison between $e^{-itH}$ and $e^{-itH_0}$ given by \eqref{eq:close-to-free} is directly useful for deriving information about the Fourier localization
of eigenfunctions of the random Schrodinger operator $H$, in the spirit of~\cite{SSW} .
To state the result, let
$\rho\in C_c^\infty([-1,1])$ be a bump function such that $\rho(x) = 1$ on $[-1/2,1/2]$, and define
$\rho_{\delta,E}(x) := \rho((x-E)/\delta)$.  Then $\rho_{\delta,E}(H)$ is a spectral projection onto eigenvalues
near $E$ of $H$, and $\rho_{\delta,E}(H_0)$ is a spectral projection onto the same window for the free Hamiltonian.
One has the following high probability estimates for approximate eigenfunctions:

\begin{restatable}{corollary}{corFourier}
\label{corFourier}
For any $K > \sqrt{\log R}$ the bound
\begin{equation}
\label{eq:projection-bound}
\|\rho_{\delta,E}(H) - \rho_{\delta,E}(H_0)\|_{\mathrm{op}} \leq C K \lambda|\log \lambda|^2 \delta^{-1/2}
\end{equation}
holds for all $\delta>0$ with probability at least $1-e^{-cK^2}$.
With the same probability, the following hold:
\begin{itemize}
\item \textbf{Fourier localization of approximate eigenfunctions:}
If $\psi\in L^2$ satisfies $\|(H-E)\psi\|_{L^2} \leq \lambda^2$ and is normalized so that $\|\psi\|_{L^2}=1$, then
\begin{equation}
\label{eq:SSW-bd}
\|\psi - \rho_{\delta,E}(H_0) \psi\|_{L^2} \leq C K \lambda|\log \lambda|^2 \delta^{-1/2}.
\end{equation}
\item \textbf{Spatial delocalization of approximate eigenfunctions:}
The estimate
\begin{equation}
	\label{eq:spatial-delocalization}
\sup_{x} ||\psi \chi_{B_\ell(x)}||_{L^2}\leq K\lambda|\log \lambda|^3\sqrt{\ell}\frac{\log(2+E)}{E^{1/4}}
\end{equation}
holds for all $\ell>0$ and any $\psi\in L^2$ satisfying $\|(H-E)\psi\|_{L^2} \leq \lambda^2$.
\item \textbf{Fourier localization of long-time evolution:} If $\psi_0 = \rho_{\delta,E}(H_0)\psi_0$ then
\begin{equation}
\label{eq:KE-localization}
\sup_{t\in\Real} \|\rho_{\delta,E}(H_0) e^{-itH} \psi_0 - e^{-itH}\psi\|_{L^2}
\leq C K |\log \lambda|^2 \lambda\delta^{-1/2}.
\end{equation}
\end{itemize}
\end{restatable}
Note that the final three points~\eqref{eq:SSW-bd}, \eqref{eq:spatial-delocalization}, and \eqref{eq:KE-localization} are deterministic consequences of~\eqref{eq:projection-bound}.
The first and second points are stated for approximate eigenfunctions since our
model has a compactly supported potential. One could obtain a statement on true eigenfunctions
by restricting to finite volume.
The third point addresses
a remark of Erd\"os, Salmhofer, and Yau~\cite{erdHos2008quantum} in which it is suggested that after time $\lambda^{-4}$ there might be transitions between  different energy shells.  At least in a scaling limit where $\lambda\to 0$ and kinetic energy is not rescaled, this is ruled out by~\eqref{eq:KE-localization}.

\subsection{Random Floquet systems}\label{sec:Time-periodic}
To highlight the robustness of our propagator based approach we also consider random potentials $V(x,t)$ which are periodic in time in the sense that
$V(x,t+\tau) = V(x,t)$ for some $\tau$. Define for each $t>0$
\begin{align*}
H_t & = H_0 + V(x,t),\\
V(x,t) & := \sum_{n\in \mathbb{Z}^d, |n|\leq R} g_n \phi_n(t)V_0(x-n),
\end{align*}
where $V_0\in L^\infty(\mathbb{R}^d)$ is a bounded function supported in $B_1$, $(g_n)_{n\in \mathbb{Z}^d}$ are independent standard Gaussians,
and for each $n\in \mathbb{Z}^d$, $\phi_n\in C^\infty(\mathbb{R})$ is a $\tau$-periodic function.  On $\bbZ^d$ we use the potential
\[
V(t) := \sum_{n\in\mathbb{Z}^d,|n|\leq R} g_n \phi_n(t) \ket{n}\bra{n}.
\]

Let $U(b,a)$ be the unitary evolution from time $a$ to time $b$ of the dynamics generated by $H$. We prove the following
analogue of Theorem~\ref{thm:thmMain} in this case.
\begin{theorem}
\label{thm:time-dependent-propagator}
Let $H$ be the $\tau$-periodic Hamiltonian defined above.  There exists a constant $c>0$
independent of $\lambda$ such that for all $K > \sqrt{\log (R)}$
\begin{equation}
\label{eq:close-free-timedep}
	\Prob(\|e^{-itH_0}- U(a+t,a)\|_{\mathrm{op}} > K \lambda|\log \lambda| \sqrt{\sigma_d(t)}
\text{ for some }t,a\in\Real) < 2\exp(-cK^2).
\end{equation}
\end{theorem}

From Theorem \ref{thm:time-dependent-propagator} we derive a modified version
of Corollary~\ref{corFourier} which applies to steady states of the Floquet system satisfying $U(\tau,0)\psi = e^{-i\mcal{E}\tau}\psi$
for a \textit{quasienergy} $\mcal{E}$ (see~\cite{sambe1973steady}).

To define a spectral projection $f(U)$ to the quasienergy we proceed
with the Fourier series as follows.  Given some function $f:\Sphere^1\to\Complex$, we write $f$ in the form
\[
	f(e^{i\theta}) = \sum_{k\in \Z} a_k e^{ik\theta},
\]
and we extend this to $U$ by
\[
f(U) := \sum_{k\in \Z} a_k U^k.
\]
Let $f_{\delta,\theta_0}\in C^\infty(\Sphere^1)$ be a smooth function supported in a $\delta$-neighborhood
of $e^{i\theta_0}$.
Note that, setting $g_{\delta,\theta_0}(\theta) = f(e^{i\theta})$, we have the identity
\begin{align*}
f_{\delta,\theta_0}(e^{i\tau H_0})
&= \sum_{k\in\bbZ} g_{\delta,\theta_0+2\pi k}(\tau H_0)  \\
&= \sum_{k\in\bbZ} g_{\tau^{-1}\delta,\tau^{-1}\theta_0+2\pi\tau^{-1} k}(H_0).
\end{align*}
That is, $f_{\delta,\theta_0}(e^{i\tau H_0})$ is a $2\pi\tau^{-1}$-periodized projection onto the spectrum of $H_0$.
In particular, in the case that $H_0$ is the Laplacian on $\bbZ^d$ (and thus has bounded spectrum) and $\tau < \pi$, we
have that $f_{\delta,\theta_0}(e^{i\tau H_0})$ is a projection onto just one interval of width $\delta\tau^{-1}$.

\begin{corollary}
\label{corTimeDependent}
For any $K>\sqrt{\log (R+\tau)}$ we have
\begin{equation}
\label{eq:U-proj-comparison}
\|f_{\delta,\theta}(e^{i\tau H_0}) - f_{\delta,\theta}(U)\|_{\mathrm{op}}
\leq C K  \lambda|\log\lambda|^2 (\delta\tau^{-1})^{-1/2}
\end{equation}
for all $\delta>0$ with probability at least $1-e^{-cK^2}$. On this event the following bounds hold:
\begin{itemize}
\item \textbf{Fourier localization of approximate eigenfunctions:}  If $\|U\psi - e^{i\tau E}\psi\| \leq \lambda^2$
then
\[
\|\psi - f_{\delta,\theta}(e^{i\tau H_0})\psi\|_{L^2} \leq C K \lambda|\log\lambda|^2 (\delta\tau^{-1})^{-1/2}.
\]
\item \textbf{Fourier localization of long-time evolution:}  If $f_{\delta,\theta}(e^{i\tau H_0})\psi_0 = \psi_0$
then for all $N\in\bbN$
\[
\|f_{\delta,\theta}(e^{i\tau H_0}) U^N \psi_0 - U^N\psi_0\|_{L^2}
\leq C K  \lambda|\log\lambda|^2 (\delta\tau^{-1})^{-1/2}.
\]
\end{itemize}
\end{corollary}
The time-independent potential is in fact periodic with all periods $\tau>0$, and indeed one can observe
that Corollary~\ref{corTimeDependent} implies Corollary~\ref{corFourier} by taking a limit $\tau\to 0$.

\subsection{Overview of the proof}
All of the above results are ultimately derived by combining standard bounds on the norm of a random matrix with dispersive estimates for the free propagator.
To explain how this works, let us consider for the sake of simplicity the time-independent model described
in~\eqref{eq:HZd},~\eqref{eq:VR-Zd} (for the model on $\bbZ^d$), or~\eqref{eq:modelRd} (on $\Real^d$).
We write
\[
H = H_0 + \lambda V,
\]
so that has the propagator has the form $U(t,0) = e^{-itH}$.

Theorem \ref{thm:thmMain} may be thought of as a ``square-root'' cancellation estimate for the $T_j$, that is, we show that that, with high probability,
\[
\|T_j(t)\| \leq \left(C \sigma_d(t)/ j\right)^{j/2},
\]
for some constant $C$ independent of $t$ and $j$.
Note that this is roughly a square-root improvement of the trivial bound
$\|T_j(t)\| \leq \frac{1}{j!}t^j\|V\|_{L^\infty}$, which follows from the triangle inequality.

At a high level, there are three ingredients to this estimate: (1) a lemma showing that square-root cancellation for $T_j(t)$ follows from square-root cancellation for $T_1(t)$, (2) an approximation to $T_1(t)$ by a structured random matrix (in the sense of~\cite{van2017structured}) of finite dimension, and (3) an application of off-the-shelf random matrix theory bounds to this finite rank approximation,
which requires as input dispersive estimates for the free evolution $e^{-itH_0}$.

For (1), we find that there is some ``approximate product structure'' which allows us to
estimate the operators $T_k$ in terms of $T_1$.  Remarkably, this argument does not rely on the randomness of $V$ and involves only an elementary geometric treatment of the simplex in \eqref{eq:Tj-def}. Roughly, one can show that
\begin{align*}
	\|T_j(t)\|_{\text{op}}\leq \|T_1(t/j)\|_{\text{op}}^j
\end{align*}
(see Lemma \ref{lem:Tone-to-Tk} below), so our problem reduces to an estimate of the operator $T_1(t)$.
The benefit of working with $T_1(t)$ is that it is a random operator that is \textit{linear} in terms of the potential $V$.
By expanding $V$ as a sum, we see that it is of the form
\[
T_1(t) = \sum_{\substack{n\in\bbZ^d \\ |n|\leq R}} g_n A_n(t),
\]
where $A_n(t)$ are the operators
\[
A_n(t) = \int_0^t e^{isH_0} \chi_n e^{-isH_0} \diff s.
\]
Written this way, we  can see that $T_1(t)$ has the structure of a (infinite-dimensional) Gaussian random matrix with nontrivial correlation between the entries.  We will show that because of the compact support of the potential, $T_1(t)$ is well approximated by finite rank operators. Therefore, its norm may be effectively bounded
using the following general-purpose lemma, which is by this point classical (see the original works of Lust-Piquard and Pisier~\cite{lust1986inegalites,lust1991non} or, for example, a modern presentation for matrices in~\cite{van2017structured}) and is proved in the appendix for the convenience of the reader.  The bound in expectation is otherwise known as the \emph{noncommutative Khintchine inequality}.

\begin{restatable*}[Matrix concentration for structured random matrices]{theorem}{NCK}
\label{thm:main-RMT}
Let $X = \sum_{n=1}^s g_n A_n$ be a $d\times d$ random matrix with random coefficients $g_n$ and the $A_n$ are Hermitian matrices. Here the $g_n$ are either independent standard Gaussian random variables or are independent variables satisfying $\Expec g_n =0$ and $|g_n|\leq 1$ almost surely.
Then there exists a universal constant $C>0$ so that
\begin{align}
\label{eq:NCKExp}
\Expec \|X\|_\mathrm{op} \leq C (\log d)^{1/2}  \sigma,
\end{align}
where
\begin{align*}
	\sigma:=\|\sum_{n=1}^s A_n^2\|_\mathrm{op}^{1/2}.
\end{align*}
Furthermore, for $K>4\sqrt{\log d} $ we have the bound
\begin{equation}
\label{eq:NCKTail}
\Prob(\|X\|_\mathrm{op} \geq K \sigma) \leq
\exp(- \frac{K^2}{10}).
\end{equation}
\end{restatable*}

The square-root cancellation phenomenon for $\|T_1(t)\|_{\text{op}}$ is reminiscent of the square root cancellation for the norm of an $d\times d$ Wigner matrix with independent entries.
This general purpose bound gives non-trivial information about the norm of such a matrix, though it is off by a factor of $\sqrt{\log d}$ for the expectation and a power of $d$ for the fluctuations.
Indeed, the tail bound \eqref{eq:NCKTail} is in general suboptimal given that $\|X\|_{\text{op}}$ should have stronger Gaussian concentration properties. As we plan to demonstrate in subsequent work, the random variable $\|T_1\|_{\text{op}}$ does in fact enjoy sharper concentration about its mean than is suggested by \eqref{eq:NCKTail}.

To apply this estimate to $T_1(t)$, all that remains is to estimate the quantity $\|\sum_n A_n^2\|_{\text{op}}$.
As it turns out, this quantity reduces to the on-diagonal decay of the propagator $e^{-itH_0}$, which
goes as $t^{-d/2}$ on both $\bbZ^d$ and $\Real^d$.\footnote{The \emph{off}-diagonal decay of the free propagator degenerates to $t^{-\frac{d}{3}}$ on $\Z^{d}$, however.}

\subsection{Related work}
\label{sec:RelatedWork}

This work fits into a larger story of attempting to understand the behavior of eigenfunctions of random Schrodinger
operators in the delocalized phase, which we make no attempt to comprehensively survey. Let us, however, mention some past works that we deem to be particularly relevant.

Ideas from random matrix theory have long played a role in this investigation, a notable
early example being the work of Magnen, Poirot, and Rivasseau~\cite{magnen1997anderson}.  In that work, the averaged Green's function was treated as a random matrix
and a renormalization group approach was introduced to understand its average.

The results of~\cite{magnen1997anderson} motivated the work of Schlag, Shubin, and Wolff~\cite{SSW}.
There, instead of renormalization techniques, more sophisticated harmonic analysis was used to show operator norm bounds on the random operator
$\rho_{\delta,E}(H_0) V \rho_{\delta,E}(H_0)$. From this, they obtained a delocalization result for \emph{most} eigenfunctions of a random Schr\"{o}dinger operator on a
finite box, uniformly in the size of the box. This is closely analogous to our estimate~\eqref{eq:SSW-bd}.
Note, however, that we instead consider a model on an infinite domain with a potential of support size $R$ and prove results for \emph{every} almost-eigenfunction, but with probability depending on $R$. While these two types of statements are not exactly the same, our methods can be used to deduce a result more comparable with that of~\cite{SSW}.
The argument of~\cite{SSW} relied on ideas from restriction theory which in particular
required a curvature assumption on the level sets of the dispersion relation, and was also limited to two dimensions.
The notable differences between our result and theirs are therefore that (1) we do not impose any restriction on the energy level, in particular allowing the flat level surface at $E=0$ and (2) our result works in any dimension $d\geq 2$.

Around the same time, while studying the scattering theory of an operator with a random decaying potential,
Bourgain in~\cite{Bourgain} showed that one can control the difference of the free and perturbed propagators in norm via an alternative argument based on the ``dual Sudakov'' entropy inequality.
Bourgain did not directly work with $\rho_{\delta,E}(H_0)V\rho_{\delta,E}(H_0)$
but rather with the Born series for the resolvent,
\[
(H-E+i\eta)^{-1} = R_0 - R_0 VR_0 + R_0VR_0VR_0 + \cdots,
\]
where $R_0 = (H_0-E+i\eta)^{-1}$ and $H=H_0+V$.
The terms of the Born series can be bounded in terms of the operator $R_0^{1/2} V R_0^{1/2}$, which he treated without the use of restriction theory, relying only on elementary Fourier analysis analogous to our estimates for the mean of $T_1(t)$.
Though he did not use restriction theory, his results still require the exclusion of certain energies to avoid singularities in the level set of the dispersion relation.
Like $\rho_{\delta,E}(H_0) V\rho_{\delta,E}(H_0)$ (and $T_1(t)$ which we consider in our paper) $R_0^{1/2}VR_0^{1/2}$ is a random operator that is linear in $V$.  Thus one could if desired replace the use of the dual Sudakov inequality in~\cite{Bourgain} by an application of the noncommutative Khintchine inequality.\footnote{In fact, both $\rho_{\delta,E}(H_0)V\rho_{\delta,E}(H_0)$ and $R_0^{1/2}VR_0^{1/2}$ have the form that the randomness factors out as a diagonal matrix so that both the results of~\cite{SSW} and~\cite{Bourgain} follow more simply from the argument in~\cite{rudelson1999random}.} While the setting treated by Bourgain is different from ours, we expect that his techniques could be adapted to obtain an energy restricted version of our Theorem \ref{thm:thmMain}.
On the other hand, his methods do not seem to apply to time-dependent potentials or other choices of $H_0$.

\subsection*{Outline of the paper} This paper is arranged as follows: in Sections \ref{sec:T1boundZd} and \ref{sec:T1boundRd} we derive tail bounds for $\|T_1\|_{\textrm{op}}$ on $\Z^{d}$ and $\R^{d}$, respectively, when $V$ is time-independent. Then in Section~\ref{sec:T1toTk} we show how the bounds for $\|T_1\|_{\textrm{op}}$ imply estimates for $\|T_j\|_{\textrm{op}}$, which gives the proof of Theorem~\ref{thm:thmMain}. In Section~\ref{sec:timeDep} we explain briefly how the proofs of the previous sections may be adapted to the time-dependent case in order to prove Theorem~\ref{thm:time-dependent-propagator}. Then in Section~\ref{sec:consequences} we derive the various corollaries of our theorems. Finally, we include two Appendices one with some background from random matrix theory and another with a construction of a finite-rank phase space localizer.

\subsection*{Notation and conventions}
We adopt the following notation and conventions throughout:
\begin{itemize}
	\item $C$ and $c$ denote positive constants that may change from instance to instance and depends only on the parameters defining $V$. In particular, it is independent of $t$, $j$, and $\lambda$.
	\item The symbol $\|\cdot\|$ always denotes the $L^2$ norm on either  $\R^{d}$ or $\Z^{d}$. Similarly, $\|\cdot\|_{\text{op}}$ denotes the $L^2\to L^2$ operator norm.
	\item For $f\in L^2(\Z^{d})$, we take its Fourier transform to be
		\begin{align*}
			\hat{f}(\xi)=(2\pi)^{-\frac{d}{2}}\sum_{n\in \Z^{d}}e^{in\xi}.
		\end{align*}
	\item For $f\in L^2(\R^{d})$, we take its Fourier transform to be
		\begin{align*}
			\hat{f}(\xi)=(2\pi)^{-\frac{d}{2}}\int_{\R^{d}}e^{-ix\xi}f(x)\:\diff x
		\end{align*}

\end{itemize}

\subsection*{Acknowledgements}
The authors would like to thank Wilhelm Schlag and Gigliola Staffilani for helpful discussions.
FH is supported by NSF DMS-2302094.
\section{$T_1$ bounds on $\Z^d$}\label{sec:T1boundZd}
In this section, we prove the following Proposition:
\begin{proposition}\label{pr:T_1Prop}
Let $H$ be as above on $\Z^{d}$, $d\geq 2$. Then there exists constants $C>0$ and $c>0$ such that for all $K>\sqrt{\log R}$
\begin{align*}
\Prob(\|T_1(t)\|_{\mathrm{op}} \geq C K \sqrt{t \log(2+t)}
\text{ for some }t>0) \leq \exp(-cK^2),
\end{align*}
if $d=2$ and
\begin{align*}
\Prob(\|T_1(t)\|_{\mathrm{op}} \geq CK\sqrt{t}
\text{ for some }t>0) \leq \exp(-cK^2),
\end{align*}
if $d\geq 3$.
\end{proposition}

We estimate $\|T_1(t)\|_{\mathrm{op}}$ differently depending on the timescale. At timescales $t\gg R^C$,
the particle escapes the compact support of the potential, so that $T_1$ may be controlled via deterministic scattering argument.
Up until that timescale, the particle can continue to interact with the potential
so we use the noncommutative Khintchine inequality to capture the cancellations from these interactions.

All of our bounds for $T_1(t)$  are based only on the following dispersive estimate for $U_0(s) = e^{-isH_0}$.
\begin{lemma}
\label{lem:nonstationary}
The propagator $U_0(s)$ satisfies the following bounds:
\begin{itemize}
\item \textbf{Local dispersive bound:}
For $n,m\in\bbZ^d$ with $|n-m| \leq \frac{1}{100} s$,
\begin{equation}
\label{eq:Zd-dispersive}
\langle n | U_0(s) |m\rangle \leq C|s|^{-d/2}.
\end{equation}
\item \textbf{Finite speed of propagation:} For $|n-m| > 10 s$, and all $N>0$
\begin{equation}
\label{eq:Zd-propagation}
\langle n | U_0(s) | m\rangle  \leq C_N t^N |n-m|^{-N}.
\end{equation}
\end{itemize}
\end{lemma}
\begin{proof}
Recall that
\begin{align*}
    \langle n|U_0(s)|m\rangle =(2\pi)^{-d}\int_{\bbT^d}e^{-i2s\sum_{j=1}^d\cos(\xi_j)-i(n-m)\cdot \xi}\:d\xi=(2\pi)^{-d}\prod_{j=1}^d\int_0^{2\pi}e^{is\phi_{(n_j-m_j)/ s}(\xi)}\:d\xi_j,
\end{align*}
for
\begin{align*}
	\phi_w(z)=-2\cos(z)-wz.
\end{align*}
The phase $\phi_w(z)$  has a critical point for $z$ such that $\sin(z)=\frac{w}{2}$, and it is non-degenerate unless $z=\frac{\pi}{2},\frac{3\pi}{2}$. The condition $|x-y|\leq \frac{1}{100}s$ precludes the degeneracy, so \eqref{eq:Zd-dispersive} follows from Van der Corput's lemma. On the other hand, the condition $|x-y|>10s$ ensures that the phase $\phi_{(x_j-y_j)/s}(\xi_j)$ is non-stationary for all $j$ so \eqref{eq:Zd-propagation} follows from integration by parts.
\end{proof}

The following establishes Proposition \ref{pr:T_1Prop} for $t> R^{5d}$.
\begin{lemma}\label{lem:long time bound}
For all $K \geq \sqrt{\log(R)}$ we have
\begin{equation}
\Prob(\|T_1\|_{\mathrm{op}}>K \log(2+t)R^{4}\text{ for some }t>0) \leq 2e^{-cK^2},
\end{equation}
if $d=2$ and
\begin{equation}
	\Prob(\|T_1\|_{\mathrm{op}}> K R^{2d}\text{ for some }t>0)\leq 2e^{-cK^2},
\end{equation}
if $d\geq 3$.
\end{lemma}

\begin{proof}
For each $N>0$ and $t>1$ we have
\begin{align*}
\|T_1(t)\|_{\mathrm{op}}^{2N}
& = \|(T_1(t)^*T_1(t))^N\|_{\mathrm{op}} \\
& \leq ||V||_{L^\infty}^{2N}\int_{0}^t\dots \int_0^t \prod_{i=1}^{2N} \|\chi_{<R} U(s_i-s_{i-1})\chi_{<R}\|_{\mathrm{op}}\,\diff s_1\dots \diff s_{2N},
\end{align*}
where $\chi_R$ is the indicator function of the ball of radius $R$ centered at the origin.
For $|t|\leq 100R$ we bound $\|\chi_{<R}U(t)\chi_{<R}\| \leq 1$ by unitarity.  Otherwise we can use~\eqref{eq:Zd-dispersive} and a lossy estimate for the operator norm of the matrix $\chi_{<R}U(t)\chi_{<R}$ in terms
of the sums of the entries to obtain
\[
	\|\chi_{<R} U(t)\chi_{<R}\|_{\text{op}} \leq R^d |t|^{-d/2}.
\]
In either case it follows that $\|\chi_{<R}U(t)\chi_{<R}\|_{\text{op}} \leq R^d(1+|t|)^{-d/2}$.
Therefore,
\begin{align*}
	\|T_1(t)\|_{\textrm{op}}^{2N}
& \leq ||V||_{L^\infty}^{2N} R^{2dN}\int_0^t\dots \int_0^t \prod_{i=0}^{2N} (1+ |s_i-s_{i-1}|)^{-d/2}ds_1\dots ds_{2N}.
\end{align*}
 Performing the integration and raising each side to the $1/2N$ power gives
\begin{align*}
    \|T_1(t)\|_{\mathrm{op}} \leq \|V\|_{\mathrm{op}} R^4 t^{1/(2N)} (C\log(2+t)),
\end{align*}
if $d=2$ and
\begin{align*}
\|T_1(t)\|_{\mathrm{op}}\leq C \|V\|_{\mathrm{op}} R^{2d} t^{1/(2N)},
\end{align*}
if $d\geq 3$. Now taking $N\to\infty$ and applying the following bound on the maximum of $R^d$ many sub-Gaussian
random variables
\begin{align}\label{eq:VProbBound}
    \mathbb{P}(\|V\|_{L^\infty} > K)\leq 2e^{-cK^2}
\end{align}
for $K > \sqrt{\log(R)}$ gives the result.
\end{proof}

Now we prove Proposition \ref{pr:T_1Prop} for $t< R^{5d}$. First we prove the following bound on $\|T_1\|_{\mathrm{op}}$ for a single time.
\begin{proposition}\label{pr:T_1-one-point}
 For all $K \geq \sqrt{\log(R)}$ and $t>0$ we have
	\begin{align}
        \Prob\left( \|T_1(t)\|_\mathrm{op}\geq K \sqrt{t\log(2+t)} \right)\leq 2e^{-cK^2}
    \end{align}
	if $d=2$ and
    \begin{align}
        \Prob\left(\|T_1(t)\|_\mathrm{op}\geq K \sqrt{t}\right)\leq 2e^{-cK^2}
    \end{align}
    if $d\geq 3$.
\end{proposition}

When $t>R^{5d}$ the claim follows immediately by Lemma \ref{lem:long time bound},
and in the case $t<1$ the naive triangle inequality bound $\|T_1(t)\|_{\text{op}} \leq t\|V\|_{\infty}$ suffices.
The content of Proposition~\ref{pr:T_1-one-point} is therefore a bound for the case $1<t<R^{5d}$.
To apply noncommutative Khintchine, we use a propagation estimate to pass to a finite-dimensional approximation of $T_1(t)$, given by $\chi_L T_1(t) \chi_L$, for $L = 4R^{10d}$.
The following lemma justifies this approximation.
\begin{lemma}
\label{lem:Ttilde-appx}
For $L = 4 R^{10d}$ and $t<R^{5d}$, we have that
    \begin{align*}
        &\|T_1(t)-\chi_L T_1(t)\chi_L\|<C R^{-10}||V||_{L^\infty}.
    \end{align*}
\end{lemma}
\begin{proof}
    First, write
    \begin{align*}
        \bra{m'}T_1(t)\ket{m}=\sum_{|n|\leq R}g_n
\int_0^t
\langle m' | U_0(-s)|n\rangle\langle n | U_0(s) | m\rangle \diff s.
    \end{align*}
If $m\not\in B_L$, then by~\eqref{eq:Zd-propagation} and the triangle inequality we have
\[
\langle m'|T_1(t)|m\rangle
\leq C \|V\|_{L^\infty} R^d t |m|^{-10d}.
\]
A similar bound holds if $m'\not\in B_L$ with $m'$ replacing $m$.
Since
\begin{align*}
    \bra{m'}T_1(t)-\tilde{T}_1(t)\ket{m}=
\begin{cases}\bra{m'}T_1(t)\ket{m}&\text{$m$ or $m'$ is in $B_L^c$}\\
    0&\text{otherwise}
    \end{cases},
\end{align*}
the desired bound now follows from the Schur test.
\end{proof}

We are now ready to complete the proof of Proposition~\ref{pr:T_1-one-point} in the case $1<t<R^{5d}$.
\begin{proof}[Proof of Proposition~\ref{pr:T_1-one-point}]
	By Lemma~\ref{lem:Ttilde-appx} and \eqref{eq:VProbBound}, it suffices to prove the analogous bound for $\tilde{T_1}(t)$.
	To do so, write $\tilde{T}_1(t)$ as
\begin{align*}
    &\tilde{T}_1(t)=\sum_{|n|\leq R}g_n A_n\\
    &A_n(t):=\chi_{L}\int_0^t e^{isH_0}\ket{n}\bra{n}e^{-isH_0}\:\diff s\chi_{L}.
\end{align*}
Now we may apply noncommutative Khintchine (Theorem~\ref{thm:main-RMT}) to $\tilde{T}_1(t)$ as a random operator acting on $\ell^2(B_{L})$ to get that
\begin{align*}
\Prob(\|\tilde{T_1}(t)\|_{\mathrm{op}}\geq K \|\sum_{|n|\leq R} A_n(t)^2\|_{\mathrm{op}}^{1/2})\leq 2e^{-cK^2}
\end{align*}
for all $K \geq \sqrt{\log(R)}$. To finish, we estimate $\|\sum_{n\in B_R} A_n(t)^2\|_{\mathrm{op}}$ as follows
\begin{equation}
\label{eq:22}
\begin{aligned}
\| \sum_{|n|\leq R} A_n^2(t)\|_{\mathrm{op}}
& \leq \int_0^{t}\int_0^{t}\|\chi_L\sum_{|n|\leq R} e^{is_1H_0}\ket{n}\braket{n|e^{-is_1H_0}\chi_Le^{is_0H_0}|n}
\bra{n}e^{-is_0H_0}\chi_L \|_{\mathrm{op}}\: \diff s_0 \diff s_1\\
& \leq \int_0^{t}\int_0^{t} \|\sum_{|n|\leq R} \braket{n|e^{-is_1H_0}\chi_L e^{is_0H_0}|n} \ket{n}\bra{n} \|_{\mathrm{op}} \: \diff s_0\diff s_1.
\end{aligned}
\end{equation}
From this, we see that it is enough to show that
\begin{align*}
	\sup_{|n|\leq R}|\braket{n|e^{-is_1H_0}\chi_L e^{is_0H_0}|n}|\leq C\left<s_1-s_0 \right>^{-\frac{d}{2}}.
\end{align*}
This follows from writing
\begin{align*}
	\braket{n|e^{-is_1H_0}\chi_L e^{is_0H_0}|n}=\braket{n|e^{-i(s_1-s_0)H_0}|n}-\braket{n|e^{-is_1H_0}\chi_L^c e^{is_0H_0}|n},
\end{align*}
and then noting that the first term is controlled by~\eqref{eq:Zd-dispersive} and the second by~\eqref{eq:Zd-propagation}. Having bounded $\sum_{|n|\leq R}A_n^2(t)$, the proof is now complete.
\end{proof}

To finish the proof of Proposition \ref{pr:T_1Prop} we use the following lemma and a union bound
\begin{lemma}
\label{lem:one-to-sup}
For each $t\geq 0$,
\begin{align*}
\|T_1(t)\|_{\mathrm{op}}\leq ||V||_{L^\infty} + \sum_{j=0}^{\lfloor\log_2(t)\rfloor}\|T_1(2^j)\|_{\mathrm{op}}.
\end{align*}
\end{lemma}
\begin{proof}
We argue by induction. The claim clearly holds for $t\geq 1$ and $t=2^k$, $k\geq 0$. Now take $t\in (2^{n_0},2^{n_0+1})$, $n_0\geq 0$, and assume
it holds for all $s\leq 2^{n_0}$. Applying the triangle inequality gives
\begin{align*}
\|T_1(t)\|_{\mathrm{op}}
& \leq \|T_1(2^{n_0})\|_{\mathrm{op}} + \|\int_{2^{n_0}}^{t} U_0(-s)VU_0(s)\:\diff s\|_{\mathrm{op}}\\
&  =  \|T_1(2^{n_0})\|_{\mathrm{op}} + \|T_1(t-2^{n_0})\|_{\mathrm{op}}.
\end{align*}
Since $t-2^{n_0}\in (0,2^{n_0})$ we may apply the inductive hypothesis to
$\|T_1(t-2^{n_0})\|_{\mathrm{op}}$ and the claim follows.
\end{proof}
The above lemma and Proposition \ref{pr:T_1-one-point} imply Proposition \ref{pr:T_1Prop} for $t<R^{5d}$ by a union bound.
Indeed, for $d=2$, we have
\begin{align*}
\Prob(\|T_1(t)\|_{\mathrm{op}}\geq K \sqrt{t\log(t)}\: \text{ for some } t<R^{5d})
& \leq \sum_{j=1}^{Cd\log(R)}\Prob(\|T_1(2^j)\|\geq \frac12 K \sqrt{2^j\log(2^j)})\\
& \leq Cd\log(R)e^{-cK^2},
\end{align*}
and similarly for $d\geq 3$.
Now taking $c>0$ sufficiently small proves Proposition \ref{pr:T_1Prop} in the range $t<R^{5d}$. Since
Lemma 2.1 implies it for $t\geq R^{5d}$ we are done.

\section{$T_1$ bounds on $\R^{d}$}\label{sec:T1boundRd}
In this section we prove the analogue of Proposition~\ref{pr:T_1Prop} on $\R^{d}$ :
\begin{proposition}\label{pr:T_1PropRd}
Let $H$ be as above on $\R^{d}$, $d\geq 2$. Then there exists constants $C>0$ and $c>0$ such that for all $K>\sqrt{\log(R)}$
\begin{align}
	\Prob\left( \|T_1(t)\|_\mathrm{op}\geq K \sqrt{t\log(2+t)}\text{ for some } t>0\right)\leq 2e^{-cK^2},
\end{align}
if $d=2$ and
\begin{align}
	\Prob\left(\|T_1(t)\|_\mathrm{op}\geq K \sqrt{t}\text{ for some }t>0\right)\leq 2e^{-cK^2}
\end{align}
    if $d\geq 3$.
\end{proposition}
The proofs of Lemmas \ref{lem:long time bound} and \ref{lem:one-to-sup} hold verbatim in $\mathbb{R}^d$. Hence, it suffices to
establish Proposition \ref{pr:T_1-one-point} in this setting. The main difference in $\mathbb{R}^d$ is that finite speed of propagation no longer holds, which means that it is not as straightforward to approximate $T_1(t)$ by a finite rank operator in order to apply the noncommutative Khintchine inequality. However, the restriction of $T_1$ to low frequencies will obey the finite speed of propagation and therefore be sufficiently compact, whereas the high frequency component will not interact with the potential much.\par
As before, for $n\in \Z^{d}\cap B_R$ we define
\begin{align*}
	A_n(t)=\int_0^{t}e^{isH_0}V_0(\cdot-n)e^{-isH_0}\:\diff s.
\end{align*}
We will prove the following:
\begin{proposition}\label{pr:T_1approxRd}
	For any $L>R$ and $\eps>0$, there exists an operator $\Pi_{L,\eps}$ of rank at most $\eps^{-2d}L^{2d}$ such that
 \begin{align}
	 \|T_1(t)-\Pi_{L,\eps}^* T_1(t)\Pi_{L,\eps} \|\leq C \|V\|_{L^\infty} (t^{1/2} L^{-1/2} R^{1/2} + t L^{-10d}+t\eps),
 \end{align}
 and
 \begin{align}\label{eq:A_napprox}
	 \|A_n(t)-\Pi_{L,\eps}^*A_n(t)\Pi_{L,\eps}\|_{\mathrm{op}}\leq C(t^{\frac{1}{2}}L^{-\frac{1}{2}}+tL^{-10d}+t\eps)
 \end{align}
\end{proposition}
Note that we can be very lossy in the rank of $\Pi_{L,\eps}$ because the operator norm estimate from the noncommutative Khintchine inequality depends only logarithmically on the dimension of the matrix.

Again for convenience we state the bounds we will use for the free Schr\"odinger propagator $U_0(s)$ on $\Real^d$.  Its
Schwartz kernel is given by
\[
\langle x|U_0(s)|y\rangle = \int e^{i(x-y)\cdot \xi} e^{it|\xi|^2/2} \diff \xi.
\]
Given a frequency cutoff $L>0$, we define the multiplier operator $P_{< L}$ by
\[
\Ft{P_{<L} \psi}(\xi) = \rho(\xi/L) \hat{\psi}(\xi),
\]
where $\rho\in C_c^\infty(\Real^d)$ is some fixed smooth function which is identically $1$ in $B_1$ and is supported in $B_2$.
We write $P_{>L} = 1-P_{<L}$.  Then for example
\begin{align}\label{eq:kernel-rep}
\langle x|U_0(s)P_{<L}|y\rangle = \int e^{i(x-y)\cdot \xi} e^{it|\xi|^2/2} \rho(\xi/L) \diff \xi.
\end{align}
We can now state the estimates we need.
\begin{lemma}
\label{lem:free-ests-Rd}
The free Schr\"odinger propagator $U_0(s)$ satisfies the following estimates.
\begin{itemize}
\item \textbf{High-frequency behavior:} For $L>0$, if $|x-y|\leq \frac{1}{100} Ls$ we have
\begin{equation}
\label{eq:disp-highfreq}
|\langle x | U_0(s) P_{>L} |y\rangle| \leq C_N (Ls)^{-N}.
\end{equation}
\item \textbf{Low-frequency behavior:} For $L>0$ if $|x-y| > 100 Ls$ we have
\begin{equation}
\label{eq:disp-lowfreq}
|\langle x | U_0(s) P_{<L} |y\rangle| \leq C_N |x-y|^{-N}.
\end{equation}
\item \textbf{Standard dispersive estimate:} For any $s>0$,
\begin{equation}\label{eq:disp-Rd}
|\langle x|U_0(s) |y\rangle| \leq C s^{-d/2}.
\end{equation}
\end{itemize}
\end{lemma}
\begin{proof}
	The inequality \eqref{eq:disp-lowfreq} follows from integrating by parts in \eqref{eq:kernel-rep}, as the phase is non-stationary. The proof of \eqref{eq:disp-highfreq} is similar whereas \eqref{eq:disp-Rd} is well-known.
\end{proof}
\par
We present the proof of Proposition~\ref{pr:T_1approxRd} via a series of lemmas. First we show that very high frequencies do not substantially interact with the potential.
\begin{lemma}
\label{lem:high-freq-cutoff}
	For all $L>R\geq 1$ and $t>0$, we have that
	 \begin{align*}
		\|T_1(t)P_{>L}\|_{\mathrm{op}}\leq C \|V\|_{L^\infty} L^{-1/2} R^{1/2} t^{1/2},
	\end{align*}
	and
	\begin{align*}
	\|A_n(t)-A_n(t)P_{>L}\|_{\mathrm{op}}\leq CL^{-\frac{1}{2}}t^{\frac{1}{2}}.
	\end{align*}
\end{lemma}
\begin{proof}
First, note that
\begin{align*}
	\|T_1(t)P_{>L}\|_{\text{op}}^2&=\|T_1(t)P_{>L}^2T_1(t)\|_{\text{op}}\\
	&\leq \int_{0}^t\int_{0}^t \|Ve^{-i(s-s')H_{0}}P_{>L}^2V\|_{\text{op}}\:\diff s \diff s' \\
&\leq \|V\|_{L^\infty}^2 \int_{0}^t\int_{0}^t \|\chi_R e^{-i(s-s')H_{0}}P_{>L}^2 \chi_R\|_{\text{op}}\:\diff s \diff s'.
\end{align*}
To estimate the operator norm of $\chi_R U_0(t) P_{>L}^2 \chi_R$ we use unitarity when $|t|$ is small and otherwise
we can use~\eqref{eq:disp-highfreq} along with the Schur test to conclude
\begin{equation}
\|\chi_R U_0(t) P_{>L}^2 \chi_R \|_{\text{op}}\leq
\begin{cases}
1, & |t| \leq 100 R/L \\
C_d R^d (Lt)^{-2d}, &|t| \geq 100R/L.
\end{cases}
\end{equation}
Applying this bound above we obtain
\[
	\|T_1(t) P_{<L}\|_{\text{op}}^2 \leq C \|V\|_{L^\infty}^2 t L^{-1} R,
\]
and the desired bound follows upon taking a square root. The bound for $A_n(t)$ now follows from the same proof with  $R=1$.
\end{proof}

Now we show that we can apply a spatial cutoff to $T_1(t)$.
\begin{lemma}
\label{lem:spatial-cutoff}
For $K>100$, $L>R$, and $t<R$ we have that
\[
\|T_1(t) P_{\leq L} (1-\chi_{KL})\|_{\mathrm{op}} \leq C t \|V\|_{L^\infty} (KL)^{-10d},
\]
and
\begin{align*}
\|A_n(t) P_{\leq L} (1-\chi_{KL})\|_{\mathrm{op}} \leq C t \|V\|_{L^\infty} (KL)^{-10d},
\end{align*}
\end{lemma}
\begin{proof}
Writing out the definition of the operator and using the triangle inequality we have
\[
	\|T_1(t) P_{\leq L} (1-\chi_{KL})\|_{\text{op}} \leq
\|V\|_{L^\infty}
\int_0^t \|\chi_R U(s) P_{\leq L} (1-\chi_{KL})\|_{\text{op}} \diff s.
\]
Now applying~\eqref{eq:disp-lowfreq} and the Schur test we have for $s<t$
\[
	\|\chi_R U_0(s) P_{\leq L} (1-\chi_{KL})\|_{\text{op}}  \leq C \int_{|z|>KL}z^{-11d}\:dz\\
	\leq C(KL)^{-10d},
\]
whereupon integrating over $t$ concludes the proof. As before, the same proof confirms the bound for $A_n$.
\end{proof}

Now set $\Pi_L = P_{\leq L} \chi_{100L}$. We need that $\Pi_L$ can be quantitatively approximated by a finite-rank operator, i.e.
\begin{restatable}{lemma}{PiFiniteRank}
\label{lem:Pi_finite_rank}
For any $0<\eps<1$ there is an operator $\Pi^\eps_L$ satisfying
\[
	\|\Pi_L - \Pi^\eps_{L}\|_{\mathrm{op}} \leq \eps
\]
and with rank at most $|\log\eps|^{2d}\eps^{-2d}L^{4d}$.
\end{restatable}
Since this is relatively standard, we postpone the proof to Appendix~\ref{sec:PhaseSpaceLocal}.
With these results, the proof of Proposition \ref{pr:T_1approxRd} is straightforward:
\begin{proof}[Proof of Proposition \ref{pr:T_1approxRd}]
We write
\begin{align*}
	\|T_1(t) - \Pi_{L,\eps}^* T_1(t) \Pi_{L,\eps} \|_{\text{op}}
&\leq \|T_1(t) - T_1(t) \Pi_{L,\eps} \|_{\text{op}} + \|T_1(t) \Pi_{L,\eps} - \Pi_{L,\eps}^* T_1(t)\Pi_{L,\eps}\|_{\text{op}}\\
&\leq 2 \|T_1(t) (1-\Pi_{L,\eps})\|_{\text{op}} \\
&\leq 2\|T_1(t)(1-\Pi_L)\|_{\text{op}}+\|T_1(t)(\Pi_L-\Pi_{L,\eps})\|_{\text{op}}.
\end{align*}
The second term is controlled by $\eps t\|V\|_{L^{\infty}}$ due to Lemma \ref{lem:Pi_finite_rank}. For the first, we combine Lemma~\ref{lem:high-freq-cutoff} and Lemma~\ref{lem:spatial-cutoff} to conclude that
\begin{align*}
 \|T_1(t) (1-\Pi_{L})\|_{\text{op}}
&\leq 2 (\|T_1(t) P_{>L} \chi_{KL}\|_{\text{op}}  + \|T_1(t) P_{\leq L}(1-\chi_{KL})\|_{\text{op}})  \\
&\leq C \|V\|_{L^\infty} (t^{1/2} L^{-1/2} R^{1/2} + t L^{-10d}),
\end{align*}
which finishes the proof, the bound for $A_n$ being essentially the same.
\end{proof}
Taking $L = R^{100d}$ we can conclude that for $|t|<R^{5d}$ (justified by Lemma~\ref{lem:long time bound}) we have
\[
	\|T_1(t) - \Pi_L T_1(t) \Pi_L\|_{\mathrm{op}} \leq C \|V\|_{L^\infty}.
\]

We are now ready to prove the main estimate on $T_1$:
\begin{proposition}\label{pr:T_1-one-point-Rd}
For all $K \geq \sqrt{\log(R)}$ and $t>0$ we have
\begin{align}\label{eq:T_1MainEst>2}
        \Prob\left( \|T_1(t)\|_\mathrm{op}\geq K \sqrt{t\log(t)} \right)\leq 2e^{-cK ^2}
    \end{align}
	if $d=2$ and
    \begin{align}\label{eq:T_1MainEst=2}
        \Prob\left(\|T_1(t)\|_\mathrm{op}\geq K \sqrt{t}\right)\leq 2e^{-cK ^2}
    \end{align}
    if $d\geq 3$.
\end{proposition}
\begin{proof}
	As explained above, it suffices to consider $1<|t|<R^{5d}$. Furthermore, with $L=R^{100d}$ and $\eps=L^{-1}$, by Proposition \ref{pr:T_1approxRd}, it suffices to establish \eqref{eq:T_1MainEst>2} and \eqref{eq:T_1MainEst=2} for $\Pi_{L,\eps}^*T_1(t) \Pi_{L,\eps}$. Recall that
	\begin{align*}
		\Pi_{L,\eps}^* T_1(t)\Pi_{L,\eps}= \sum_{n\in B_R\cap \Z^{d}}g_n \Pi_{L,\eps}^*A_n(t)\Pi_{L,\eps},
	\end{align*}
	which is of the correct form to apply Theorem \ref{thm:main-RMT}. Since the dimension of each $\Pi_{L,\eps}^*A_n(t)\Pi_{L,\eps}$ is logarithmic in $\lambda$, the result will follow from the noncommutative Khintchine inequality if we can show that
\begin{align}\label{eq:NCKVarParam}
		\left\|\sum_{n\in B_R\cap \Z^{d}}\Pi_{L,\eps}^*A_n(t)\Pi_{L,\eps}\Pi_{L,\eps}^*A_n(t)\Pi_{L,\eps}\right\|_{\text{op}}\leq
		C\begin{cases}
			t&d\geq 3\\
			t\log(2+t)&d=2.
		\end{cases}
\end{align}
For this, we write
\begin{align*}
		&\left\|\sum_{n\in B_R\cap \Z^{d}}\Pi_{L,\eps}^{*}A_n(t)\Pi_{L,\eps}\Pi_{L,\eps}^*A_n(t)\Pi_{L,\eps}\right\|\leq\\
		&\quad\left\|\sum_{n\in B_R\cap \Z^{d}}A_n(t)^2\right\|+\left\|\sum_{n\in B_R\cap\Z^{d}}\Pi_{L,\eps}^{*}A_n(t)\Pi_{L,\eps}\Pi_{L,\eps}^*A_n(t)\Pi_{L,\eps}-A_n(t)^2\right\|.
\end{align*}
The second term may be treated by writing
\begin{align*}
	&\left\|\sum_{n\in B_R\cap\Z^{d}}\Pi_{L,\eps}^{*}A_n(t)\Pi_{L,\eps}\Pi_{L,\eps}^*A_n(t)\Pi_{L,\eps}-A_n(t)^2\right\|\leq\\
	&\quad \sum_{n\in B_R\cap \Z^{d}}\left\|\Pi_{L,\eps}^{*}A_n(t)\Pi_{L,\eps}(\Pi_{L,\eps}^*A_n(t)\Pi_{L,\eps}-A_n(t))\right\|
	+\sum_{n\in B_R\cap \Z^{d}}\left\|(\Pi_{L,\eps}^*A_n(t)\Pi_{L,\eps}-A_n(t))A_n(t) \right\|,
\end{align*}
and applying Proposition~\ref{pr:T_1approxRd} since the volume factor incurred by the sum may be absorbed by the strong decay in \eqref{eq:A_napprox}. For the first term, we use the definition of $A_n(t)$ and unitarity to write
\begin{align}\label{eq:A_nSumInt}
\left\| \sum_{n\in B_R\cap \Z^{d}}A_n(t)^2\right\|_{\text{op}}\leq \int_0^t\int_0^t\left\|\sum_{n\in B_R\cap \Z^{d}} V_0(\cdot-n)e^{i(s-s')H_0}V_0(\cdot-n)\right\|_{\text{op}}\:\diff s \diff s'.
\end{align}
Now, because $\supp V_0\subset B_1$, for each $V_n$, $\supp V_n\cap \supp V_m=\emptyset$ for all but  $O(2d)$ many $m$. Therefore, we may write for any $f\in L^2$
\begin{align*}
\left\|\sum_{n\in B_R\cap \Z^{d}} V_0(\cdot-n)e^{i(s-s')H_0}V_0(\cdot-n)f\right\|^2
&\leq C \sum_{n \in B_R\cap \Z^{d}} \|V_ne^{i(s-s')H_0}V_n\|_{\textrm{op}} \|V_nf\|^2\leq \\
&\leq C \|f\|^2\sup_{n\in B_R\cap \Z^{d}}\|V_ne^{i(s-s')H_0}V_n\|_{\text{op}}.
\end{align*}
The inequality \eqref{eq:NCKVarParam} now follows from noting that
\begin{align*}
	\|V_ne^{i(s-s')H_0}V_n\|_{\text{op}}&\leq \|V_n\|_{L^\infty\to L^2}\|e^{i(s-s')}\|_{L^1\to L^{\infty}}\|V_n\|_{L^{2}\to L^{1}}\\
	&\leq C \left<s-s' \right>^{-\frac{d}{2}},
\end{align*}
by \eqref{eq:disp-Rd}, and then integrating in \eqref{eq:A_nSumInt}. With this, the proof is complete.
\end{proof}

\section{From $T_1$ to $T_k$}\label{sec:T1toTk}
\label{sec:TonetoK}

In this section, we prove the following deterministic lemma allowing control of higher $T_k$ in terms of $T_1$.
As a consequence, we derive Theorem~\ref{thm:thmMain} from Propositions~\ref{pr:T_1Prop} and \ref{pr:T_1PropRd}.

\begin{restatable}{lemma}{ToneToTk}
\label{lem:Tone-to-Tk}
For each $t> 0$, define
\begin{equation}\label{eq:T1bound}
	M_t:= \sup_{s \in [0,t]}\frac{\|T_1(s)\|_{\mathrm{op}}}{\sqrt{s}}.
\end{equation}
There exists an absolute constant $C>0$ such that for all $k$ and $t>0$
\begin{equation}
\label{eq:first-Tk-bd}
\|T_k(t)\|_{\mathrm{op}} \leq (C k^{-1/2}(||V||_{L^{\infty}} + M_t)\sqrt{t\log(2+t)})^{k}.
\end{equation}
\end{restatable}
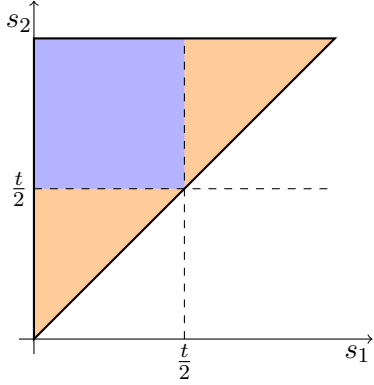
\begin{figure}
\begin{tikzpicture}
\draw[->] (-0.2, 0) -- (4.5, 0) node[right] {};
\draw[->] (0, -0.2) -- (0, 4.5) node[above] {};

\node at (2, -0.3) {$\frac{t}{2}$}; \node at (-0.2,2) {$\frac{t}{2}$};
\node at (4.3, -0.2) {$s_1$}; \node at (-0.2,4.2) {$s_2$};
\fill[fill=blue!30] (0,2) rectangle (2,4);
\fill[orange!40] (0,0) -- (2,2) -- (0,2) -- cycle;
\fill[orange!40] (2,2) -- (4,4) -- (2,4) -- cycle;
\draw[dashed] (0,2) -- (4,2);
\draw[dashed] (2,0) -- (2,4.0);
\draw[thick] (0,0) -- (4,4) -- (0,4) -- cycle;
\end{tikzpicture}
\caption{An illustration of the product structure identity for $k=2$.  The operator $T_2(t)$ is
given by an integral of $V(s_2)V(s_1)$ for $(s_1,s_2)$ in the shaded triangle.  In $(s_1,s_2)$
space, translation corresponds to intertwining the free evolution on either side of $V(s_2)$ and $V(s_1)$, since $V(s+s') = e^{is'H_0} V(s) e^{-is'H_0}$.  Thus the contribution of the blue square
corresponds to a portion of the integral in $T_2(t)$ that has operator norm bounded by
$\|T_1(t/2)\|^2$.  Similarly the two orange triangles are have operator norms bounded by $\|T_2(t/2)\|$.  Iterating this construction allows one to bound $\|T_2(t)\|$ in terms of $\|T_1(2^{-j}t)\|^2$.}
\label{fig:T2}
\end{figure}

The proof will follow from the following the observations about the operator $T_j$:
\begin{lemma}
\label{lemTkalg}
	The operators $T_j$ satisfy the following:
\begin{itemize}
\item \textbf{A priori bound:}  For any $t$,
\begin{equation}
\label{eq:Tj-apriori}
\|T_j(t)\|_{\textrm{op}} \leq \|V\|_{L^\infty}^j \frac{t^j}{j!}.
\end{equation}
\item \textbf{Product structure:} For $s,t\geq 0$ we have
\begin{equation}\label{eq:Tprod}
	T_j(s+t) = \sum_{k=0}^k T_k(s) T_{j-k}(t).
\end{equation}
\item \textbf{Decomposition:}
	\begin{align}\label{eq:Tdecomp}
	T_j(t) = \sum_{\substack{\vec{m} \in \bbN_{\geq 0}^j \\ \sum m_k = j}}
	T_{m_j}(t/j) T_{m_{j-1}}(t/j) \cdots T_{m_1}(t/j).
\end{align}
\end{itemize}
\end{lemma}
\begin{proof}
	The a priori bound comes from the putting the norm inside the integral in \eqref{eq:Tj-def} and the fact that the simplex
	\begin{align*}
		\Delta_j(t):=\{(s_1,s_2,\ldots,s_k)\in \R_{\geq 0} \mid s_1\leq s_2\leq\ldots\leq s_k\leq t\}
	\end{align*}
has volume $\frac{t^j}{j!}$.\par
The product structure comes from writing
\begin{align*}
	\Delta_j(t+s)=\bigsqcup_{k=0}^j \{(s_1,\ldots,s_j)\in \Delta_j(s+t) \mid s_1,\ldots,s_k\leq s\}
\end{align*}
and noting that
\begin{align*}
	&\int_{\{(s_1,\ldots,s_j)\in \Delta_j(s+t) \mid s_1,\ldots,s_k\leq s\} }e^{i(t+s)H_0}V(s_j)V(s_{j-1})\cdots V(s_1)\:d\vec{s}\\
	&\quad =\int_{s\leq s_{k+1}\leq\ldots\leq s_j\leq t+s}e^{i(t+s)H_0}V(s_j)\cdots V(s_{k+1})e^{-isH_0}\:d\vec{s} T_k(s)\\
	&\quad = T_{j-k}(t)T_k(s).
\end{align*}
Finally, the decomposition identity follows by writing $t=\frac{t}{j}+(j-1)\frac{t}{j}$ and applying the product structure identity iteratively.
\end{proof}

We are now ready to prove Lemma~\ref{lem:Tone-to-Tk}.
\begin{proof}[Proof of Lemma~\ref{lem:Tone-to-Tk}]
We first prove it for $j=2$. When $t<1$ the claim follows from (\ref{eq:Tj-apriori}) so we assume $t\geq 1$ and iterate (\ref{eq:Tdecomp}) with $j=2$  and apply (\ref{eq:Tj-apriori})
to get that for any $N\geq 1$
\begin{align*}
\| T_2(t)\|_{\mathrm{op}}
& \leq 2\| T_2(t/2)\|_{\mathrm{op}} + \|T_1(t/2)\|_{\mathrm{op}}^2 \\
& \leq 2^N \|T_2(t/2^N)\|_{\mathrm{op}} + \sum_{k=1}^N 2^{k-1}\|T_1(t/2^k)\|_{\mathrm{op}}^2\\
& \leq ||V||_{L^\infty}^2t^2/2^N + N M_1^2(t) t.
\end{align*}
Choosing $N$ such that $t/2 < 2^N < 2t$ gives the claim. Now we show by induction that for all $j\geq 1$
\begin{equation}
	\label{eq:M1-estimate}
\|T_j(t)\|_{\mathrm{op}} \leq (C_j j^{-1/2}(||V||_{L^\infty}+M_t)\sqrt{t\log(2+t)})^j
\end{equation}
where $C_j = 8\max_{0<k<j} k^{-1/2}C_k$.

Suppose for some $j>2$ (\ref{eq:M1-estimate}) holds for all $k<j$. Then let $\bar{C} = \max_{1\leq k<j}C_k k^{-1/2}$, use (\ref{eq:Tdecomp}) and the induction hypothesis to estimate
\begin{align*}
\|T_j(t)\|_{\mathrm{op}}
& \leq \sum_{\substack{\vec{m} \in \bbN_{\geq 0}^j \\ \sum m_k = j}}
\prod_{j=1}^k \|T_{m_j}(t/k)\|_{\textrm{op}} \\
& \leq  j\|T_{j} (t/j)\|_{\mathrm{op}} + 4^j (\bar{C}(||V||_{L^{\infty}} + M_t)\sqrt{t/j\log(2+t/j)})^j.\\
& \leq j\| T_{j} (t/j)\|_{\mathrm{op}} + (4\bar{C} (||V||_{L^{\infty}} + M_t) j^{-1/2}\sqrt{t\log(2+t)})^j.
\end{align*}
If we iteratively apply this estimate we obtain that for any $N\geq 0$
\begin{align*}
\|T_j(t)\|_{\mathrm{op}}
& \leq  j^N \| T_{j}(t/j^N)\|_{\mathrm{op}} + (\sum_{k=0}^{N-1} j^{k}j^{-jk/2}) (4\bar{C} (||V||_{L^{\infty}} + M_t) j^{-1/2}\sqrt{t\log(2+t)})^j.
\end{align*}
To finish, we take $N\to\infty$, use the a-priori bound (\ref{eq:Tj-apriori}) on the first term, and note that since $j>2$
the sum in the second term on the RHS converges proving (\ref{eq:M1-estimate}).
\end{proof}

With Lemma~\ref{lem:Tone-to-Tk} in hand, the proof of Theorem~\ref{thm:thmMain} follows immediately:
\begin{proof}[Proof of Theorem~\ref{thm:thmMain}]
From Propositions~\ref{pr:T_1Prop} and \ref{pr:T_1PropRd}, we have that for all $K>\sqrt{\log R} $ and $t>0$
\begin{align*}
	M_t = \sup_{s\in[0,t]} \frac{\|T_1(s)\|}{\sqrt{s}} \leq
\begin{cases}
CK, &d>2 \\
CK \sqrt{\log(2+t)}, &d=2.
\end{cases}
\end{align*}
with probability at least $1-e^{-cK^2}$. Using that $\|V\|_{L^\infty}\leq K $ with probability at least $1-e^{-cK^2}$, \eqref{eq:mainT_jEst} now follows immediately from \eqref{eq:first-Tk-bd}.
\end{proof}

We are now ready to prove the remaining conclusions of Corollary~\ref{cor:close-to-free}
and its analogue for time-dependent systems, Theorem~\ref{thm:time-dependent-propagator}.  The proof of
Theorem~\ref{thm:time-dependent-propagator} is exactly analogous, just with the notational modifications described
in Section~\ref{sec:timeDep}.  For simplicity we write the proof for the time-independent case below.
\begin{proof}[Proof of Corollary \ref{cor:close-to-free}]
	From Theorem \ref{thm:thmMain}, we have that for all $K>\sqrt{\log R}$
\begin{align*}
	\|T_j(t)\|_{\text{op}}\leq \left(CK^2\frac{\sigma_d(t)}{j}\right)^{j / 2},
\end{align*}
holds with probability at least $1-e^{-cK^2}$ so that
\begin{align*}
\|e^{-itH} - e^{-itH_0}\|_{\textrm{op}}
&\leq \sum_{j=1}^\infty \lambda^j \|T_j(t)\|_{\textrm{op}} \\
&\leq \sum_{j=1}^\infty (CK^2\lambda^2 \sigma_d(t) /j)^{j/2}.
\end{align*}
If $CK^2\lambda^2\sigma_d(t)\leq \frac{1}{2}$, then the first term dominates for a bound of
\begin{align*}
\|e^{-itH} - e^{-itH_0}\|_{\textrm{op}}\leq C K\lambda \sigma_d^{1 /2}(t).
\end{align*}
Combining this with the trivial bound $\|e^{-itH} - e^{-itH_0}\|_{\textrm{op}}\leq 2$ yields \eqref{eq:close-to-free}

To prove the convergence of the Dyson series~\eqref{eq:truncation} we write
\begin{equation*}
\|e^{-itH} - \sum_{j=0}^{M} T_j(t)\|_{\textrm{op}} \leq
\sum_{j=M+1} \lambda^j (C K^2 \frac{\sigma_d(t)}{j})^{j/2}.
\end{equation*}
Assuming that $M / 10>  A :=C K^2 \lambda^2\sigma_d(t) $, the series is monotone decreasing so we may estimate it by
\begin{align*}
\int_{M}^\infty e^{-\frac{x}{2}(\log x-\log A)}\:dx\leq e^{-M},
\end{align*}
as desired.
\end{proof}

\section{Modifications for the time-dependent potential}\label{sec:timeDep}
\label{sec:floquet}
In this section we describe the modifications one needs to make to the above proof of Theorem~\ref{thm:thmMain}
to obtain Theorem~\ref{thm:time-dependent-propagator}.  For simplicity we work on $\Real^d$, the case of $\bbZ^d$ being essentially the same.

Recall that the time-dependent system is defined by the Schrodinger equation
\[
i\partial_t \psi(t,x) = (H_0 \psi)(t,x) + \lambda V(t,x) \psi(t,x),
\]
where we recall that
\[
V(x,t) := \sum_{n\in\bbZ^d; |n|\leq R} g_n \phi_n(t) V_0(x-n),
\]
for $V_0 \in L^\infty(\Real^d)$ supported in $B_1$ and $\phi_n \in L^\infty(\Real)$ a bounded $\tau$-periodic function,
and $g_n$ independent bounded or Gaussian random variables.
For notational convenience we write $V^t$ for the potential at time $t$.
We write $U(b,a)$ for the unitary evolution operator mapping an initial data at time $a$ to the solution at time $b$.
For uniformity of notation we also write $U_0(b,a) := e^{-i(b-a) H_0}$ for the free evolution operator.
The Duhamel formula for the time-dependent evolution is given by
\[
U(b,a) = U_0(b,a) + \int_a^b U(b,s) V^s U_0(s,a) \diff s.
\]
Writing $V(s;a) := U_0(a,s) V^s U_0(s,a)$, we can then write out the Dyson series
\[
U(b,a) = U_0(b,a) + \sum_{j=1}^\infty U_0(b,a) T_j(b,a)
\]
where
\begin{equation}
\label{TKba-def}
T_k(b,a) := \int_{a\leq t_1\leq \cdots\leq t_k\leq b} V(t_k;a)V(t_{k-1};a)\cdots V(t_1;a) \diff \vec{t}.
\end{equation}

The proof of Proposition~\ref{pr:T_1-one-point} (either on $\bbZ^d$ or the analogue Proposition~\ref{pr:T_1-one-point-Rd})
go through verbatim to prove the following slightly more general estimate:
\begin{proposition}[Time-dependent analogue of Proposition~\ref{pr:T_1-one-point}]
 For all $\lambda \geq \sqrt{\log(R)}$, $a\in\Real$ and $t>0$ we have
	\begin{align}
        \Prob\left( \|T_1(a+t,a)\|_\mathrm{op}\geq \lambda \sqrt{t\log(t+1)} \right)\leq 2e^{-c\lambda ^2}
    \end{align}
	if $d=2$ and
    \begin{align}
        \Prob\left(\|T_1(a+t,a)\|_\mathrm{op}\geq \lambda \sqrt{t}\right)\leq 2e^{-c\lambda ^2}
    \end{align}
    if $d\geq 3$.
\end{proposition}

In the time independent case, we used Proposition~\ref{pr:T_1-one-point} to obtain a bound valid for all $t$ with high probability
using the inequality (see Lemma~\ref{lem:one-to-sup})
\[
	\|T_1(t)\|_{\textrm{op}} \leq \|V\|_{L^\infty} + \sum_{j=0}^{\lfloor \log_2(|t|)\rfloor} \|T_1(2^j)\|_{\textrm{op}}.
\]
Exactly the same argument proves the following generalization for any $m<t$:
\[
	\|T_1(t,m)\|_{\textrm{op}} \leq \|V\|_{L^\infty} + \sum_{j=0}^{\lfloor \log_2(t-m)\rfloor} \|T_1(2^j+m,m)\|_{\textrm{op}}.
\]
Then observing that
\[
T_1(t,s) = T_1(t,m) - T_1(s,m)
\]
we conclude that for any $m < a < b$,
\[
	\|T_1(b,a)\|_{\textrm{op}} \leq 2 \|V\|_{L^\infty} + 2\sum_{j=0}^{\lfloor \log_2(b-m)\rfloor} \|T_1(2^j+m,m)\|_{\textrm{op}}.
\]
Using the periodicity $T_1(b+k\tau,a+k\tau) = T_1(b,a)$, we can find $\bar{a}\in[0,\tau]$ so that $a-\bar{a}\in\tau\bbZ$ and
take $m= \lfloor \bar{a}\rfloor$, so that
\begin{align*}
	\|T_1(b,a)\|_{\textrm{op}} &= \|T_1(b-a + \bar{a}, \bar{a})\|_{\textrm{op}} \\
				   &\leq 2 \|V\|_{L^\infty} + 2\sum_{j=0}^{\lfloor \log_2((b-a) + 1)\rfloor} \|T_1(2^j+\lfloor \bar{a}\rfloor,\lfloor \bar{a}\rfloor)\|_{\textrm{op}}.
\\&\leq 2 \|V\|_{L^\infty} + 2\sum_{j=0}^{\lfloor \log_2((b-a)+1)\rfloor}
\max_{m\in [0,\tau]\cap \bbN} \|T_1(2^j + m,m)\|_{\textrm{op}}.
\end{align*}
Now by a union bound,
\begin{align*}
	\Prob(\|T_1(a+t,a)\|_{\textrm{op}} \geq \lambda \sqrt{t\log(t)} &\text{ for some } 0<a<\tau, 0<t<R^{5d})\\
&\leq \sum_{j=1}^{Cd \log(R)} \sum_{m=0}^{\lfloor \tau\rfloor}
\Prob(\|T_1(2^j+m,m)\|_{\textrm{op}} \geq \lambda \sqrt{2^j \log(2^j)}) \\
&\leq Cd \log(R) (\tau+1) e^{-c\lambda^2}.
\end{align*}
For $\lambda \gg \sqrt{\log(R+\tau)}$ the term in the
front can be absorbed up to modifying $c$.  Thus we obtain
\begin{proposition}[Time-dependent analogue of Proposition~\ref{pr:T_1Prop}]
For $T_1(b,a)$ as defined above we have
\begin{equation}
\label{eq:floquet-Tone}
\mathbb{P}(\sup_{a>0, t>0}\frac{\|T_1(a+t,a)\|_{\mathrm{op}}}{\sqrt{t+1}\log(1+t)}\geq \lambda) \leq 2\exp(-c\lambda^2),
\end{equation}
for all $\lambda > C\sqrt{\log(R+\tau)}$.
\end{proposition}

To derive a bound  $T_k(b,a)$ using the above estimate for $T_1(b,a)$ we can apply the same deterministic result
Lemma~\ref{lem:Tone-to-Tk}.  The only modification is that we need to slightly modify Lemma~\ref{lemTkalg} to incorporate
the time dependence.
\begin{lemma}[Analogue of Lemma~\ref{lemTkalg}]
The operators $T_k(b,a)$ defined in~\eqref{TKba-def} satisfy
\begin{itemize}
\item \textbf{A priori bound:}  For any $t$, the triangle inequality implies
\begin{equation}
\|T_k(s,t)\|_{\textrm{op}} \leq \big(\sup_{s'\in[s,t]} \|V_{s'}\|_{L^\infty}^k\big) \frac{(t-s)^k}{k!}.
\end{equation}
\item \textbf{Product structure:} For $s,t\geq 0$ we have
\begin{equation}
	T_k(a,b) = \sum_{j=0}^k T_j(a,c) T_{k-j}(c,b).
\end{equation}
\item \textbf{Decomposition:}
\begin{align}
	T_k(0,t) = \sum_{\substack{\vec{m} \in \bbN_{\geq 0}^k \\ \sum m_j = k}}
	T_{m_k}(0,t/k) T_{m_{k-1}}(t/k,2t/k) \cdots T_{m_1}((k-1)t/k, t).
\end{align}
\end{itemize}
\end{lemma}
The proof is exactly analogous to the proof of Lemma~\ref{lemTkalg}.  From this we obtain a bound on $T_k(b,a)$
exactly analogous to Lemma~\ref{lem:Tone-to-Tk}.  Then Theorem~\ref{thm:time-dependent-propagator}
follows from summing the Dyson series.

\section{Spectral and dynamical consequences}
\label{sec:consequences}
In this section we derive Corollary~\ref{corFourier} as a consequence of~\eqref{eq:close-to-free}.
We give the proof for time independent and time dependent potentials on $\mathbb{R}^d$, the proof on $\mathbb{Z}^d$ being completely analogous.
Let $\rho\in C_c^\infty(\Real)$ be a smooth bump function satisfying $\rho(x)=1$ for $|x|\leq \frac{1}{2}$ and
$\rho(x)=0$ for $|x|>2$. Let $\rho_{E,\delta}(x) := \rho((x-E)/\delta)$.
The proof of Corollary 1.3 uses the fact that for any self-adjoint operator $A$ we have the relation
\begin{equation}
\label{eq:spectral-projection}
\rho_{E,\delta}(A)
= \delta \int_{\mathbb{R}} e^{it(A-E)} \Ft{\rho}(\delta t) \diff t.
\end{equation}
via the Fourier inversion formula. Note the integral in $\eqref{eq:spectral-projection}$ is absolutely
convergent since $\Ft{\rho}(t)$ has rapid decay for $|t|\gg 1$ as $\rho$ is smooth.

\begin{proof}[Proof of Corollary~\ref{corFourier} using Corollary~\ref{cor:close-to-free}]

By~\eqref{eq:close-to-free} we have the bound
\[
\|e^{itH} - e^{itH_0}\|_{\textrm{op}} \leq C K \lambda|\log \lambda|^2  t^{1/2}
\]
for all $t$ with probability at least $1-e^{-cK^2}$.
Combining the above with  ~\eqref{eq:spectral-projection} gives
\begin{equation*}
\begin{split}
\|\rho_{E,\delta}(H) - \rho_{E,\delta}(H_0)\|_{\textrm{op}} &\leq \delta \int_{\mathbb{R}} \|e^{it H} - e^{itH_0}\|_{\textrm{op}} \Ft{\rho}(\delta t) \diff t\\
&\leq C K \lambda|\log \lambda|^2  \delta \int_{\mathbb{R}}  t^{1/2} \Ft{\rho}(\delta t) \diff t \\
&\leq C K \lambda|\log\lambda|^2 \delta^{-1/2},
\end{split}
\end{equation*}
where we used the smoothness of $\rho$ to pass to the last line. This proves \eqref{eq:projection-bound}.  Now the consequences~\eqref{eq:SSW-bd}, (\ref{eq:spatial-delocalization}),
and~\eqref{eq:KE-localization} follow quickly.

\noindent
\textit{Proof of Fourier localization~\eqref{eq:SSW-bd}:}
Suppose $\psi \in L^2(\Real^d)$ satisfies $\|(H-E)\psi\| \leq \lambda^2$ and $\|\psi\|=1$.
First note that by \eqref{eq:spectral-projection} and the hypothesis on $\psi$ we have
\begin{align*}
||\psi - \rho_{E,\delta}(H)||_{L^2}
& = \delta ||\int_{\mathbb{R}}(1-e^{it(H-E)})\psi \hat{\rho}(\delta t)dt||_{L^2}\\
& \leq C\lambda^2\delta^{-1}.
\end{align*}
Hence we can estimate
\begin{align*}
\|\psi - \rho_{E,\delta}(H_0)\psi\|_{L^2}
&\leq \|\psi - \rho_{E,\delta}(H)\psi\| + \|(\rho_{E,\delta}(H) - \rho_{E,\delta}(H_0)) \psi\|_{L^2}  \\
&\leq C\lambda^2\delta^{-1} + CK |\log \lambda|^2 \lambda \delta^{-1/2},
\end{align*}
where we applied \eqref{eq:projection-bound} to bound the second term.
This implies the \eqref{eq:SSW-bd} if $\lambda^2\leq \delta$, and if $\lambda^2\geq \delta$,
the claim is trivial since we always have $\|\psi - \rho_{E,\delta}(H_0)\psi||_{L^2}\leq 2.$

\textit{Proof of Spatial Delocalization (\ref{eq:spatial-delocalization}):}
Supoose $\psi\in L^2(\mathbb{R}^d)$ satisfies $\|(H-E)\psi\|_{L^2} \leq \lambda^2$ and $\|\psi\|_{L^2} = 1$. Fix any $x_0\in \mathbb{R}^d$, $\ell\geq 1$.
If $\ell>\lambda^{-2}E^{1/2}$, then claim is trivial, so we assume in what
follows that $\ell\leq \lambda^{-2} E^{1/2}$.

Take $N>0$ such that $2^{N-1}\lambda^2 \leq \ell^{-1} \leq 2^N\lambda^2$ and
decompose $\psi$ in Fourier space by
\begin{align*}
\psi
& = \rho_{E,\lambda^2}(H_0)\psi + \sum_{k=0}^N (\rho_{E,2^{k+1}\lambda^2}(H_0) -  \rho_{E,2^{k}\lambda^2}(H_0))\psi + (\psi - \rho_{E,2^N\lambda^2}(H_0)\psi)\\
& := \psi_0+\psi_1+...+\psi_N.
\end{align*}
Now, letting $\phi:\mathbb{R}^d\to \mathbb{R}$ be a function such that $\supp(\hat{\phi})\subset B_{C\ell^{-1}}(0)$,
$\chi_{B_{\ell}(0)}\leq \phi \leq C$ and $||\phi||_{L^2}\leq C\ell^{d/2}$, we estimate
\begin{equation}
\label{eq:frequency decomposition}
\begin{split}
||\psi \chi_{B_\ell(x_0)}||_{L^2}
& \leq \sum_{k=0}^N ||\psi_k \phi(\cdot-x_0)||_{L^2}\\
& \leq C \sum_{k=0}^{N} ||\psi_k||_{L^2} \min(1,\lambda E^{-1/4}2^{k/2})
\end{split}
\end{equation}
To pass to the second inequality we used that for $k = 0,...,N-1$ we have
$$\sup(\hat{\psi_k})\subset \{\xi\in \mathbb{R}^d\ | |\xi|^2-E|\leq C 2^k \lambda^2\}:= A_{k},$$
so we can estimate
\begin{align*}
||\psi_k \phi(\cdot-x_0)||_{L^2}
& = \int_{\mathbb{R}^d} |\int_{\mathbb{R}^d} \hat{\psi_k}(\xi-\eta)e^{-i\eta\cdot x_0}\hat{\phi}(\eta)d\eta|^2d\xi\\
& \leq (\sup_{\xi}\int_{\mathbb{R}^d}\chi_{A_k}(\xi-\eta) \chi_{B_{C\ell^{-1}}}(\eta) d\eta)\int_{\mathbb{R}^d}\int_{\mathbb{R}^d}|\hat{\psi_k}(\xi-\eta)|^2|\hat{\phi}(\eta)|^2d\eta d\xi\\
& \leq C \min(\ell^{-d}, \ell^{-d+1}\frac{2^k\lambda^2}{\sqrt{E}}) ||\psi_k||_{L^2}^2 ||\phi||_{L^2}^2\\
& \leq C \min(1, \ell \frac{2^k\lambda^2}{\sqrt{E}})||\psi_k||_{L^2}^2.
\end{align*}
To finish we note that by (\ref{eq:SSW-bd}) we have
\begin{equation}
\label{eq:application of frequency concentration}
||\psi_k||_{L^2}\leq \min(K2^{-k/2}|\log(\lambda)|^2,1)
\end{equation}
for all $k\geq 0$ with probability at least $1-e^{-cK^2}$. Plugging this bound into (\ref{eq:frequency decomposition}) and
using that $N\leq C(|\log \lambda| + \log(2+E)|)$ gives
$$||\psi_k \phi(\cdot-x_0)||_{L^2} \leq K  \lambda \sqrt{\ell}\frac{\log(2+E)}{E^{1/4}}|\log \lambda|^3,$$
with probability at least $1-e^{-cK^2}.$

\textit{Proof of Kinetic energy Localization~\eqref{eq:KE-localization}:}
If $A$ is any operator that commutes with $H$ then we have
\[
\rho_{E,\eps}(H)A - A \rho_{E,\eps}(H) = 0.
\]
Hence for any such operator we also have
\[
\|\rho_{E,\eps}(H_0)A - A \rho_{E,\eps}(H_{0})\|_{\textrm{op}}
\leq 2 \|A\|\|\rho_{E,\eps}(H) - \rho_{E,\eps}(H_0)\|_{\textrm{op}}.
\]
The claim now follows immediately if we take $A = e^{-itH}$ and invoke
\eqref{eq:projection-bound}.

\end{proof}

We now turn to the time-dependent setting.
\begin{proof}[Proof of Corollary~\ref{corTimeDependent}]
Let $U_0 = e^{-i\tau H_0}$ and let $f_{\delta,\theta_0} \in C^\infty(\Sphere^1)$ be a smooth function satisfying
\[
\int_0^{2\pi} |f^{(k)}(e^{i\theta})|\diff \theta \leq C_k \delta^{1-k},
\]
so that in particular $f_{\delta,\theta_0}$ can be expressed as a Fourier series
\[
f_{\delta,\theta_0}(z) = \sum_{\bbZ} a_k z^k
\]
with $|a_k| \lsim \delta (1 + \delta |k|)^{-10}$.  We of course have in mind that $f_{\delta,\theta_0}$ is a smooth cutoff
to the interval of width $\delta$ centered at $e^{i\theta_0}$.
Then
\begin{align*}
	\|f_{\delta,\theta} (U_0) - f_{\delta,\theta}(U)\|_{\textrm{op}}
&\leq C\sum_{k\in\bbZ} \delta (1+\delta|k|)^{-10}\|U^k - U_0^k\|_{\textrm{op}}  \\
&\leq C \sum_{k\in\bbZ} K (\log\lambda^{-1})^2 \lambda (k\tau)^{1/2} \delta (1+\delta|k|)^{-10} \\
&\leq C K (\log\lambda^{-1})^2 \lambda \tau^{1/2}\delta^{-1/2}.
\end{align*}
This proves~\eqref{eq:U-proj-comparison}, and the Fourier localization bounds follow in exactly the same way as in the proof of Corollary~\ref{corFourier}.
\end{proof}

\appendix

\section{The noncommutative Khinthine inequality}
\label{sec:NCK}
The presentation in this section closely follows Section 3.1 of van Handel's survey~\cite{van2017structured}.  We
reproduce the argument here only for the purposes of completeness and convenience to the reader.  A reader interested in a more
thorough understanding of the noncommutative Khintchine inequality including illuminating examples, applications, and an explanation of the role
of commutativity / noncommutativity is encouraged to read~\cite{van2017structured}.

The noncommutative Khintchine inequality concerns matrices of the form
\begin{equation}
\label{eq:random-X}
X = \sum_{j=1}^s g_j A_j,
\end{equation}
where $g_j$ are independent Gaussians or bounded mean-zero random variables and $A_j$ are symmetric random $n\times n$ matrices.

The noncommutative Khintchine inequality states that $\|X\|_{\textrm{op}}$ is usually bounded by $\|\Expec X^2\|_{\textrm{op}}$ up to a factor of $\sqrt{\log n}$,
which is an operator version of the square root cancellation expected if for example the $A_j = a_j\Id$.  (The factor $\sqrt{\log n}$
can be seen to be necessary in the case that $X$ is a diagonal matrix with independent entries).
\NCK

Theorem~\ref{thm:main-RMT} is proven by estimating moments of $X^{2p}$.  That is,
\begin{proposition}[Non-commutative Khintchine inequality]
\label{prp:nck-moments}
For matrices $X$ of the form~\eqref{eq:random-X},
\begin{equation}
(\Expec \trace[X^{2p}])^{1/2p} \leq \sqrt{2p-1} \Big(\trace[(\Expec X^2)^p]\Big)^{1/2p}
\end{equation}
\end{proposition}

\begin{proof}[Proof of Theorem~\ref{thm:main-RMT} using Proposition~\ref{prp:nck-moments}]
For an $n\times n$ symmetric matrix $X$, we have
\begin{equation*}
	\|X\|_{\textrm{op}}^{2p} \leq \trace[ X^{2p}] = \sum_{j=1}^n \sigma_j^{2p} \leq n \|X\|_{\textrm{op}}^{2p}.
\end{equation*}
Taking $2p$-th roots we obtain $\|X\|_{\textrm{op}} \simeq \trace[X^{2p}]^{1/2p}$ for $p \simeq \log n$.
Applying Theorem~\ref{thm:main-RMT} directly with $p=\log n$ we obtain~\eqref{eq:NCKExp}.

Now we prove the concentration inequality~\eqref{eq:NCKTail}.  For simplicity we set
$B = \sum_{j=1}^s A_j^2$.
We use $\|X\|_{\textrm{op}}^{2k} \leq \trace X^{2k}$ and apply Theorem~\ref{thm:main-RMT} as follows:
\begin{equation}
\begin{split}
\Prob(\|X\|_{\textrm{op}} \geq K \|B\|_{\textrm{op}})
&\leq K^{-2p} \|B\|_{\textrm{op}}^{-p}
\Expec(\|X\|_{\textrm{op}}^{2p}) \\
&\leq K^{-2p} \|B\|_{\textrm{op}}^{-p}
\Expec(\trace[X^{2p}]) \\
&\leq K^{-2p} \|B\|_{\textrm{op}}^{-p}
(2p)^p \trace[B^p] \\
&\leq \exp\Big(-2p\log K+p\log(2p)+\log n\Big).
\end{split}
\end{equation}
Setting $p=\frac{K^2}{2e}$ and using that $K^2>16 \log n$, we find that the expression inside the parentheses is less than $-\frac{K^2}{10}$, as desired.
\end{proof}

To prove \ref{prp:nck-moments}, we first show that it suffices to establish Proposition \ref{prp:nck-moments} when the $g_j$ are Gaussian.

\begin{lemma}[Reduction to Gaussians]
Suppose that Proposition~\ref{prp:nck-moments} holds for matrices of the form \eqref{eq:random-X} when the
$g_i$ are independent standard Gaussians. Then Proposition~\ref{prp:nck-moments} holds when the $g_i$ are any collection of independent mean $0$ random variables
satisfying $|g_i|\leq 1$ almost surely.
\end{lemma}

\begin{proof}
Let $g_i$ be a family of bounded mean $0$ random variables satisfying $|g_i|\leq 1$ almost surely,
and let $h_i$ be a family of i.i.d. mean $0$ unit Gaussians.

First we show the lemma when $h_i$ are symmetric. Let $\epsilon_i$ be iid, mean $0$, $\pm 1$ valued
random variables. Using symmetry of $h_i$ and Jensens inequality, we estimate
\begin{align*}
\Expec \trace(X^{2p})
& = C^{2p} \Expec_{\epsilon}\Expec_{g} \trace((\sum_{i=1}^n \epsilon_i g_i \Expec_{g}|h_i| A_i)^{2p})\\
& \leq C^{2p} \Expec_{g} \Expec_{\epsilon}\Expec_{h} \trace((\sum_{i=1}^n \epsilon_i |h_i| g_i A_i)^{2p}),
\end{align*}
where $\Expec_{g}, \Expec_{h}, \Expec_{\epsilon}$ denote the expectation with respect to the families of random variables $g_i$, $h_i$, and $\epsilon_i$ respectively.
Since $\epsilon_i|h_i|$ are iid, centered, unit Gaussian we may apply Theorem \eqref{prp:nck-moments} and estimate
\begin{align*}
\Expec \trace(X^{2p})
& \leq C^{2p}\Expec_{g} \trace( (\sum g_i A_i^2)^{p})\\
& \leq C^{2p} \trace((\sum A_i^2)^p),
\end{align*}
where we used that the $A_i^2$ are positive definite and $|g_i|\leq 1$ almost surely. This
proves the claim when $h_i$ are symmetric.

When the $g_i$ are not symmetric we can reduce to the symmetric case as follows.
Let $\tilde{g_i}$ be an independent random variables have the same distribution
as $g_i$, and estimate using Jensens inequality
\begin{align*}
\Expec \trace(X^{2p})
& = \Expec_{g} \trace(\sum_{i=1}^n (g_i-\Expec_{\tilde{g}}\tilde{g_i})A_i)^{2p}\\
& \leq \Expec_g\Expec_{\tilde{g}} \trace(\sum_{i=1}^n (g_i-\tilde{g_i})A_i)^{2p}\\
&\leq 2^{2p} \Expec_{g,\tilde{g}} \trace(\sum_{i=1}^n \frac{1}{2}(g_i-\tilde{g_i})A_i)^{2p}.
\end{align*}

For all $i$, $g_i-\tilde{g_i}$ is symmetric and satisfies $\frac{1}{2}|g_i-\tilde{g_i}|\leq 1$ almost surely
and so the argument above for the symmetric case implies the claim.
\end{proof}
By the above lemma it suffices to prove the case when $g_i$ are unit Gaussians, and we will assume they are for the
remainder of the proof. The key point of the proof of Proposition~\ref{prp:nck-moments} is the following bound, which can be interpreted as saying that the
worst case is that $A_j$ all commute.
\begin{lemma}
\label{lem:trace-ineq}
For any symmetric matrices $A$ and $B$, and $0\leq \ell\leq 2p-2$
\begin{equation}
\trace[A B^{2p-2-\ell} A B^\ell] \leq \trace[A^2 B^{2p-2}].
\end{equation}
\end{lemma}
\begin{proof}
We can always rotate into a basis so that $B$ is diagonal with entries $B_{jk} = b_j \delta_{jk}$.  Then the trace
can be estimated using Holder's inequality as follows:
\begin{equation}
\begin{split}
\trace[A B^{2p-2-\ell} A B^\ell] &=
\sum_{j,j'} b_j^\ell b_{j'}^{2p-2-\ell} |a_{j,j'}|^2 \\
&\leq \Big(\sum_{j,j'} |b_j|^{2p-2} |a_{j,j'}|^2\Big)^{\ell/(2p-2)}
\Big(\sum_{j,j'} |b_{j'}|^{2p-2} |a_{j,j'}|^2\Big)^{(2p-2-\ell)/(2p-2)}.
\end{split}
\end{equation}
Now we recognize each sum in the product above as being equal to $\trace[A B^{2p-2} A] = \trace[A^2 B^{2p-2}]$.
\end{proof}

The other ingredient we need is the matrix version of Holder's inequality.  As a first step towards this we prove the
following bound.
\begin{lemma}[Matrix Jensen's inequality]
Let $A$ be a symmetric matrix and let $\varphi:\Real\to\Real$ be a convex function.  Then
\begin{equation}
\sum_{j=1}^d \varphi(a_{jj})  \leq \trace[\varphi(A)].
\end{equation}
\end{lemma}
\begin{proof}
Write $A = Q \Lambda Q^*$ where $\Lambda$ is diagonal, and $q_{jk}$ are the entries of $Q$.  Then
\[
a_{jj} = \sum_k \lambda_k |q_{jk}|^2.
\]
Then by Jensen's inequality using that $\sum_k |q_{jk}|^2 = 1$ for all $j$,
\begin{equation}
\begin{split}
\sum_j \varphi(a_{jj})
&= \sum_j \varphi\Big( \sum_k \lambda_k |q_{jk}|^2\Big) \\
&\leq \sum_j  \sum_k \varphi(\lambda_k) |q_{jk}|^2 \\
&= \sum_k \varphi(\lambda_k) = \trace[\varphi(A)].
\end{split}
\end{equation}
\end{proof}

\begin{lemma}[Matrix Holder's inequality]
\label{lem:holders-ineq}
For symmetric matrices $A$ and $B$ and $\frac{1}{p} + \frac{1}{p'} = 1$ with $p>1$,
\begin{equation}
\trace[AB] \leq \trace[|A|^p]^{1/p} \trace[|B|^{p'}]^{1/p'}.
\end{equation}
\end{lemma}
\begin{proof}
Without loss of generality we can take $B$ to be diagonal.  Then
\begin{equation}
\trace[AB] = \sum_{j} b_{jj} a_{jj}
\end{equation}
The result now follows from Holder's inequality, since $\sum |a_{jj}|^p \leq \trace[|A|^p]$.
\end{proof}

\begin{proof}[Proof of Proposition~\ref{prp:nck-moments}]
We use the Gaussian integration by parts formula $\Expec gf(g) = \Expec f'(g)$ for Gaussians $g$ to compute
\begin{equation}
\begin{split}
\Expec[\trace[X^{2p}]]
&= \sum_{j=1}^s \Expec[ g_j \trace[A_j X^{2p-1}]] \\
&= \sum_{\ell=0}^{2p-2} \sum_{j=1}^s \Expec[\trace[A_j X^{2p-2-\ell}A_jX^{\ell}]] \\
&\leq (2p-1) \Expec [\trace[(\sum_{j=1}^s A_j^2) X^{2p-2}]] \\
&\leq (2p-1) \trace[ (\Expec X^2)^p]^{1/p} \trace[\Expec X^{2p}]^{1-\frac1p}.
\end{split}
\end{equation}
In the last step we have used the matrix Holder's inequality and the identity $\sum_{j=1}^s A_j^2 = \Expec X^2$.  The
proof now follows from rearranging.
\end{proof}
\section{Finite-rank approximation of a phase-space localizer}\label{sec:PhaseSpaceLocal}
In this section we prove Lemma~\ref{lem:Pi_finite_rank}. There are of course many
ways to do this, below we present a simple trick to do so.  Define the harmonic oscillator Hamiltonian $Q$ by
\[
Q = \hat{p}^2 + \hat{x}^2
\]
wher $\hat{x}$ is multiplication by $x$ (the position operator) and $\hat{p} = i\nabla$ is the momentum operator,
and we write $\hat{x}^2$ for $\sum_j \hat{x}_j^2$ and likewise $\hat{p}^2 = \sum_j \hat{p_j}^2$.
The operator $e^{-sQ}$ has two useful representations that we exploit.  The first is that $e^{-sQ}$ is the time evolution operator
for the differential equation
\[
\partial_t f = -Q f.
\]
The second is the spectral representation
\[
e^{-sQ} = \sum_k e^{-s\lambda_k} \ket{\psi_k}\bra{\psi_k},
\]
where the sum is over the eigenfunctions $\psi_k$ of $Q$ with eigenvalues $\lambda_k$.
Fortunately the spectrum of $Q$ is exactly known (this is why we chose it), and in particular satisfies
\begin{equation}
\label{eq:eigencount}
\#\{k \mid \lambda_k \leq E\} \leq CE^{2d}.
\end{equation}
\PiFiniteRank*
\begin{proof}
First we show that
\[
	\|(1 - e^{-sQ}) \Pi_L \|_{\textrm{op}} \leq sL^2.
\]
Indeed,
\[
\frac{d}{ds} e^{-sQ}\Pi_L = -e^{-sQ} Q \Pi_L,
\]
so
\[
	\|(1 - e^{-sQ}) \Pi_L \|_{\textrm{op}} \leq s \|Q\Pi_L\|_{\textrm{op}}.
\]
To estimate the right hand side we write
\[
	\|Q\Pi_L\|_{\textrm{op}} \leq \|\Delta P_{\leq L} \chi_{100L}\|_{\textrm{op}} + \|\hat{x}^2 P_{\leq L} \chi_{100L}\|_{\textrm{op}}.
\]
The first term is obviously bounded by $L^2$ since $\|\Delta P_{\leq L}\|_{\textrm{op}} \leq L^2$.
For the second term we write $P_{\leq L} = \rho(\hat{p}/L)$ and use the following commutator identity:
\begin{align*}
	[\hat{x}^2,F(\hat{p})]=2i\nabla F(\hat{p})\cdot \hat{x} - \Delta F(\hat{p}),
\end{align*}
so that
\[
	\|\hat{x}^2 P_{\leq L} \chi_{100L}\|_{\textrm{op}}
\leq 2 L^{-1} \| (\nabla \rho)(\hat{p}/L)\cdot \hat{x} \chi_{100L}\|
+ L^{-2} \|(\nabla \rho) (\hat{p}/L) \chi_{100L}\| + \|P_{\leq L} \hat{x}^2 \chi_{100L}\|.
\]
Then using that $\|\hat{x}^2 \chi_{100L}\|_{\textrm{op}} \lsim L^2$ we conclude that indeed $\|Q\Pi_L\|_{\textrm{op}} \lsim L^2$.
By setting $s=s_0 := \eps L^{-2}$, we obtain
\[
	\|\Pi_L - e^{-s_0 Q}\Pi_L\|_{\textrm{op}} \leq \eps.
\]
To cutoff $e^{-s_0Q}$ to a finite rank operator we consider the operator
\[
T_{s_0,\eps} := \sum_{e^{-s_0 \lambda_k} \geq \eps} e^{-s_0\lambda_k} \ket{\psi_k}\bra{\psi_k}.
\]
We have by definition $\|e^{-s_0Q} - T_{s_0,\eps}\|_{\textrm{op}} \leq \eps$, and to count the rank of $T_{s_0,\eps}$ we observe
that $e^{-s_0\lambda_k} \geq \eps$ if $\lambda_k \leq s_0^{-1} \log\eps^{-1}$.  Then by~\eqref{eq:eigencount} we conclude
that the rank of $T_{s_0,\eps}$ is bounded by
\[
|\log\eps|^{2d} \eps^{-2d} L^{4d},
\]
as desired.
\end{proof}

\printbibliography

\end{document}